\documentclass[%
twocolumn,
showkeys,
 amsmath,amssymb,
 aps,
pra,
floatfix,
]{revtex4-2}

\newcommand{\und}{\underline}
\usepackage{graphicx,bbm}% Include figure files
\usepackage{dcolumn}% Align table columns on decimal point
\usepackage{CJK}

\usepackage{amsthm}
\usepackage{hyperref}

\usepackage{tikz}

\usepackage{float} % to fix tables

\newcounter{myequation}
\makeatletter
\@addtoreset{equation}{myequation}
\makeatother

\theoremstyle{plain}%
\newtheorem{theorem}{Theorem}%  meant for continuous numbers
\newtheorem{proposition}[theorem]{Proposition}% 
\newtheorem{lemma}[theorem]{Lemma}%  meant for continuous numbers
\newtheorem{corollary}[theorem]{Corollary}%  meant for continuous numbers

\theoremstyle{remark}%

\theoremstyle{definition}%
\newtheorem{definition}{Definition}%

\DeclareMathOperator{\supp}{supp}

\DeclareMathOperator{\dist}{dist}

\newcommand{\di}{\partial}

\DeclareMathOperator{\Tr}{Tr}

\newcommand{\g}{\gamma}
\newcommand{\al}{\alpha}

\newcommand{\Cb}{\mathbb{C}}

\newcommand{\Rb}{\mathbb{R}}
\newcommand{\N}{\mathbb{N}}

\newcommand{\om}{\omega}

\renewcommand{\l}{\lambda} % renewed from slash l
\newcommand{\abs}[1]{\ensuremath{\left\lvert#1\right\rvert}}
\newcommand{\norm}[1]{\ensuremath{\left\lVert#1\right\rVert}}
\newcommand{\sbr}[1]{\left[#1\right]}
\newcommand{\Set}[1]{\left\{#1\right\}}

\newcommand{\md}[6]{\ensuremath{
		\ifinner
		\tfrac{\partial{^{#2}}#1}{\partial{#3^{#4}}\partial{#5^{#6}}}
		\else
		\tfrac{\partial{^{#2}}#1}{\partial{#3^{#4}}\partial{#5^{#6}}}
		\fi
}}
\newcommand{\del}[1]{\left(#1\right)}
\newcommand{\thmref}[1]{Theorem~\ref{#1}}

\newcommand{\lemref}[1]{Lemma~\ref{#1}}

\newcommand{\figref}[1]{Figure~\ref{#1}}
\newcommand{\corref}[1]{Corollary~\ref{#1}}

\usepackage[normalem]{ulem}
\usepackage{soul}
\setstcolor{blue}
\definecolor{green}{rgb}{0.0, 0.5, 0.5}

\definecolor{lgray}{gray}{0.9}
\definecolor{llgray}{gray}{0.95}
\definecolor{lllgray}{gray}{0.975}

\newcommand{\cB}{\mathcal{B}}
\newcommand{\cD}{\mathcal{D}}

\newcommand{\cF}{\mathcal{F}}
\newcommand{\cG}{\mathcal{G}}

\newcommand{\cJ}{\mathcal{J}}
\newcommand{\cK}{\mathcal{K}}
\newcommand{\cL}{\mathcal{L}}
\newcommand{\cM}{\mathcal{M}}         %                      %
 
\newcommand{\cX}{\mathcal{X}}

\newcommand{\nc}{\newcommand}

\nc{\h}{\delta}
\nc{\G}{\Gamma}
\nc{\et}{\eta} 
\nc{\gam}{\gamma}
\nc{\ka}{\kappa}
\nc{\lam}{\lambda}
\nc{\Lam}{\Lambda}
\nc{\ta}{\tau}
\nc{\w}{\omega}
\nc{\io}{\iota}
\nc{\s}{\sigma}
\nc{\vphi}{\varphi}

\nc{\e}{\epsilon}

\renewcommand{\k}{\kappa}

\nc{\ran}{\rangle}
\nc{\lan}{\langle}

\newcommand{\tr}{\mathrm{Tr}}

\nc{\bfone}{{\bf 1}}
\nc{\1}{{\bf 1}}

\nc{\dd}{\mathrm{d}}
\newcommand{\DETAILS}[1]{}

\newcommand{\astlo}{\hat\chi}

\renewcommand{\cp}{\mathrm{c}}

\DeclareMathOperator{\dG}{\mathrm{d}\Gamma}
\newcommand{\SD}{\mathrm{SD}}

\newcommand{\Norm}[1]{\norm{#1}}

\DeclareMathOperator{\WC}{WC}

\begin{document}

\preprint{APS/123-QED}

	\title[Propagation of information]{Information propagation in long-range quantum many-body systems }
	
	\author{Marius Lemm}
		\email{marius.lemm@uni-tuebingen.de}
	\affiliation{Department of Mathematics, University of T\"ubingen, 72076 T\"ubingen, Germany }
		\author{Carla Rubiliani}
				\email{carla.rubiliani@uni-tuebingen.de}
	\affiliation{Department of Mathematics, University of T\"ubingen, 72076 T\"ubingen, Germany }
	\author{Israel Michael Sigal}
\affiliation{Department of Mathematics, University of Toronto, Toronto, ON M5S 2E4, Canada }
\email{im.sigal@utoronto.ca}

\author{Jingxuan Zhang}
\affiliation{Yau Mathematical Sciences Center, Tsinghua University, Beijing 100084, China }
\affiliation{Department of Mathematical Sciences, University of Copenhagen, Copenhagen 2100, Denmark}
\email{jingxuan@tsinghua.edu.cn}

\date{December 17, 2023}

\begin{abstract}
We study general lattice bosons with long-range hopping and long-range interactions decaying as $|x-y|^{-\alpha} $ with $\alpha\in (d+2,2d+1)$. We find a linear light cone for the information propagation starting from suitable initial states. We apply these bounds to estimate the minimal time needed for quantum messaging, for the propagation of quantum correlations, and for quantum state control. The proofs are based on the ASTLO method (adiabatic spacetime localization observables). Our results pose previously unforeseen limitations on the applicability of fast-transfer and entanglement-generation protocols developed for breaking linear light cones in long-range and/or bosonic systems.
\end{abstract}

\maketitle

The finite speed of information propagation is an empirical fact of Nature. In relativistic theory, the existence of the light cone is a fundamental requirement with wide-ranging consequences. It is remarkable that similar effective ``light'' cones also constrain the non-relativistic quantum theory that governs condensed-matter physics. %only quantum theory with a solid mathematical foundation and physical consistency. %background.
The existence of such an effective ``light'' cone was discovered by Lieb and Robinson (\cite{LR}) 50 years ago in quantum spin systems. %Their Lieb-Robinson bounds control the spreading of quantum correlations through spacetime bounds on the commutators of spatially localized observables.%. More precisely, \cite{LR}, showed the existence of the light cone for quantum spin lattice systems by proving space-time bounds 
%for quantum spin lattice systems  as space-time bounds   on the commutators of observables with disjoint space-time supports. % (see \eqref{} for a precise formulation).
%\vspace{3cm}

In the earliy 2000s, starting with the work of Hastings (\cite{H1}) on  the Lieb-Schultz-Mattis theorem, %in  condensed matter physics, multifarious applications of LRB to various areas of quantum physics to were discovered and a large area of theoretical and mathematical physics sprung out. Starting with the work of Hastings (\cite{H1}) on  the Lieb-Schultz-Mattis theorem and %its extension by   and
several works leveraged Lieb-Robinson bounds to great effect to control ground state correlation properties and study topological quantum phase transitions \cite{BH,BHV,EisOsb,HastKoma,H2,H3,HW,NachOgS,NachS}. These important developments revolutionized our understanding of the information content of quantum matter at zero temperature.
%followed by Nachtergaele and Sims (\cite}) on exponential decay of correlations
%in  condensed matter physics,  %continuing with works of  
%Bravyi, Hastings and Verstraete (\cite{BHV}),  Bravyi and Hastings (\cite{BH}),    Eisert and Osborn  (\cite{EisOsb}) and Hastings (\cite{H2, H3}) on quantum messaging, correlation creation, scaling,  area laws for the entanglement entropy and belief propagation  in Quantum Information Science (QIS),
% and of \cite{CL,KS_he,RS} in high energy physics,  %  In particular, important results were obtained in quantum information science (QIS), condensed matter and high energy physics, and more. %to name a few major advances.
%FLS1:  The interest in ... rapidly surged in the early 2000s when it became clear
%it transpired that Lieb-Robinson bounds (LRBs) are among the very few effective and general tools that are available for analyzing quantum many-body systems.
Since then, inspired by these works, an active research area %of theoretical and mathematical physics %related to 
dealing with dynamics of quantum information %and grounded in LRB
has sprung to life.  A variety of improvements and extensions of the original LRB have been found, e.g., to long-range spin interactions~\cite{Fossetal,TranEtAl1, TranEtAl3,TranEtal4}, lattice fermions \cite{MR3869135,WH}, open quantum lattice systems \cite{NachVerZ,Pou}, and anomalous transport \cite{DLLY,DLLY2}  have been obtained. Moreover, the applications of LRBs and related propagation bounds have been expanded and deepened to include, e.g.,  %\cite{DLLY1,DLLY2,EMNY,Fossetal,GL,GNRS,HSS,H2,MKN,NRSS,NS1,NSY1,Tranetal}
quantum state transport \cite{EpWh, FLS2} and error bounds on quantum simulation algorithms, e.g., \cite{PhysRevLett.115.130401,MR3457926,TranEtAl1, TranEtAl3},  equilibration times \cite{GogEis} in condensed matter physics,  and scrambling times relevant to high-energy physics (\cite{CL,KS_he,RS}).   %have been achieved over the past 10 years % For a more complete discussion, 
See the survey papers \cite{CLY,KGE, NachS2,NSY2} for further background on LRBs and \cite{cheneau2012light,Cheneau} for their experimental observation.

The present paper focuses on lattice bosons where it had been a long-standing problem to derive useful propagation bounds for general initial states because interactions are effectively unbounded. This problem has recently seen rapid progress \cite{FLS, FLS2,KS2,KVS,YL,WH} especially for nearest-neighbor boson hopping \cite{KVS}. Nonetheless, open problems have remained particularly concerning propagation in systems of bosons with long-range hopping and long-range interactions. Long-range interactions are subtle in quantum spin systems --- see \cite{TranEtAl1} for the phase diagram for LRBs in long-range quantum spin systems and \cite{Bose,	DefenuEtAl} for  reviews of the effect of the long-range interactions  on the transmission of quantum information and the relation to the bosonic case is subtle \cite{CGS}. For bosons, the long-range hopping and long-range interactions are effectively unbounded and the challenge is to derive LRBs with performance guarantees that are independent of ad-hoc truncation of the local particle number.

%   the existence of the infinite volume dynamics for a wide class of interaction (\cite{NachOgS, NachSchlSSZ, NachVerZ})
%A conceptual step from QM to Statistical Mechanics, i.e. from the zero density quantum systems to those of positive density was made in \cite{FLS, FLS2}.

%Concerning applications to QIT of interest in this paper, ... %\DETAILS{Not finished?}

%The next big step took almost 40 years. It was shown by  Hastings, ...  (\cite{BHV, MR2231689}) and EWh (\cite{EpWh})  that the information-theoretic operations such as  quantum messaging, creation and propagation of correlations and transport of states 
%The fact that the information-theoretic operations such as  sending messages, creation and propagation of correlations and transport of states could be derived from space-time (in this case, LR) bounds was discovered by Hastings, ...  (\cite{BHV, MR2231689}) and EWh (\cite{EpWh}), respectively. 

To tackle this problem, we consider a broad class of many-boson Hamiltonians with \textit{long-range hopping and long-range interactions} on lattices. We present new results on the minimal time for quantum messaging and the propagation of quantum correlations  %our results imply 
  (including non-existence of fast	scrambling, cf.~\cite{KS_he,CL}) % \footnote{{The scrambling time could be defined as the time an initially uncorrelated subsystem stays uncorrelated (i.e.~the time  needed to create correlations). }In particular, our results imply non-existence of fast scramling, c.f.~\cite{KS_he,CL}.}) 
 and control of states. %discrete
%  \footnote{From the condensed matter physics viewpoint, such systems arise in the standard tight-binding approximation, see e.g.~\cite{FeffLThWein} and, for rigorous results, \cite{AshMerm, Harr}. We consider them to avoid inessential technicalities in the proof of the approximation result, \thmref{thmLCA}.}. 
Concretely, fixing %a subset $\Lam$ of
  a finite lattice $\Lambda\subset \Rb^d,\,d\ge1$, we consider many-boson Hamiltonians of the form
\begin{align}\label{H} %{H\equiv} 
H_\Lam &=  \sum_{x,y \in \Lambda} h_{xy} a_x^*a_y + \frac12\sum_{\substack{x,y\in \Lam}}a^*_xa^*_y w_{xy}a_ya_x,
\end{align} 
acting on the bosonic Fock spaces $\cF_\Lam$ over the $1$-particle space $\ell^2(\Lam)$. Here $a_x$ and $a^*_x$ are bosonic annihilation and creation operators, respectively, and $h_{xy}$ is a Hermitian $|\Lambda|\times |\Lambda|$ matrix representing a $1$-particle Hamiltonian $h$ (where $|\Lambda|$ denotes the number of vertices of $\Lambda$) and $w_{xy}$ is a real-valued $|\Lambda|\times |\Lambda|$ matrix representing a $2$-particle pair potential. The results we present here extend to fermionic lattice systems.

The class of Hamiltonians \eqref{H} goes beyond Bose-Hubbard type Hamiltonians because the hopping matrix $h_{xy}$ and the two-particle interaction $w_{xy}$ can be of infinite range. %which allows to incorporate a variety of physical setup relevant to experimental realizations of lattice bosons in experiment \cite{BlochDalibardZwerger}.
	Our propagation estimates take into account the decay properties of the hopping matrix and the interaction through moments of the position operator. Namely, we assume that there exists some integer $p\geq 1$ such that
\begin{align}\label{k-cond}
	&\kappa_p=\norm{h}_p+\norm{w}_{p}<\infty,\\
	\nonumber	&\textnormal{where } \norm{u}_p=\sup _{x\in\Lam} \sum_{y\in\Lam} \abs{u_{xy}}\abs{x-y}^{p+1}.
\end{align} 
For example, in the short-range (e.g. finite-range or subexponentially decaying) regime, condition \eqref{k-cond} holds for arbitrarily large $p\ge1$. Importantly, condition \eqref{k-cond} holds also for long-range interactions $\abs{h_{xy}},\abs{w_{xy}}\leq C (1+|x-y|)^{-\alpha}$ with decay rate $\alpha>d+2$ for any $p<\alpha-d-1$.

%We remark that our results extend to long-range quantum spin Hamiltonians by cutting off the local particle number and to fermionic lattice Hamiltonians by restricting to the even CAR subalgebra when commutators are considered

% and allow $v_{xy}$ to be arbitrary (apart from the standing assumption $v_{xy}=\overline{v_{xy}}=v_{yx}$).
%The parameter $n\ge1$ in \eqref{k-cond} determines the time-/ space-decay rate in various statements.
%A sufficient condition for \eqref{k-cond} is $\abs{h_{xy}}\le C(1+\abs{x-y})^{-(n+d+2)}$ and similarly for $v$. 

%\newpage

Our first main result is the maximal velocity bound (MVB), \textbf{\thmref{thmMVB}} (i), which holds for completely general initial states. It essentially says that even bosons with long-range hopping and long-range interactions cannot propagate particles super-ballistically.  We also extend an idea of \cite{FLS2} to bound the speed of propagation of macroscopic particle clouds in \textbf{\thmref{thmPT}}. 

Next, we use the MVB to derive  \textbf{Theorem \ref{thmLCA}} on  the light-cone approximation of quantum dynamics, which, in turn, yields the weak LRB,  \textbf{\thmref{thmLRB}} for localized initial states. It is remarkable that a linear light cone for information propagation can be proved for a subclass of initial states for such general long-range bosonic Hamiltonians because such a linear light cone is expected to break for bosonic Hamiltonians with only nearest-neighbor hopping in general \cite{KVS}. Another reason the result is surprising is that for long-range spin interactions with decay rates $\alpha \in(d+2,2d+1)$ there exist fast-transfer and entanglement-generation protocols \cite{KS1,EldredgeEtAl,TranEtal4,TranEtal5} that break any linear light cone. Nonetheless, our results yield a linear light cone for such decay rates $\alpha \in(d+2,2d+1)$. Hence, our \thmref{thmLRB} places unforeseen constraints on the applicability of these fast-transfer and entanglement-generation protocols for localized initial states. We emphasize that our results are not in contradiction to these fast-transfer protocols because the Hamiltonians
we treat in \eqref{H} only involve boson density-density interactions and
we require the initial state to be suitably localized. %Put differently, we show that when all the bosons are initially concentrated, even long-range and unbounded bosonic interaction cannot propagate information super-ballistically contrary to what one might reasonably expect.

%Theorems \ref{thmLCA} and \ref{thmLRB} lead readily to Theorems \ref{thm5}--\ref{thm:gap} which give bounds on  propagation/creation of correlation, quantum messaging, state control times,  and the relation between a spectral gap and the decay of correlations,.  The proof of  \thmref{thmMVB} is then extended to yield \thmref{thmPT}  on macroscopic particle transport.

%Our main results are given in Theorems \ref{thmMVB}--\ref{thmPT} below. Theorems \ref{thmMVB},   \ref{thmLCA} and \ref{thmLRB} deal with the a maximal velocity bound (MVB),  the light-cone approximation of quantum dynamics and  the LRB. They provide different expressions of the fact that the many-body evolution stays, up to small leaking probability tails, within a light cone of the support of the initial conditions.   
We then turn to applications in  \textbf{Theorems \ref{thm5}--\ref{thm:gap}}. These results %emerge from  applications of Theorems \ref{thmLCA} and \ref{thmLRB} to
  provide general constraints on  propagation/creation of correlation, quantum messaging, state control times, and the relation between a spectral gap and the decay of correlations for the broad class of lattice boson Hamiltonians we consider. These physically meaningful consequences can be derived from the main results Theorems \ref{thmLCA} and \ref{thmLRB} through by now standard arguments  \cite{BHV,H0,H1,HastKoma,NachS,NachOgS} highlighting the ``unreasonable effectiveness of Lieb-Robinson bounds''; see also \cite{CLY,KGE, NachS2,NSY2}. 
Of specific interest is \textbf{Theorem \ref{thm5}} for pure states which bounds the minimal time for \textit{creation of quantum entanglement} between different regions. 
% In connection with our bounds on correlations, we introduce the notion of \textit{weakly correlated states} within given spatial domains, see \secref{sec:setup}. %{sec:cor}.  

%The constraints imposed by the LRB on the propagation of correlations were first discussed in, with rigorous results for fermionic systems given in \cite{}. 	
%\DETAILS{Using the LRB, the authors of \cite{NachOgS} proved an upper bound  on the rate of propagation/creation of correlations between observables with separated spatial support for long-range interacting fermions, provided the initial state is a product state (uncorrelated).}
%
%The relation between spectral gap and the decay of correlations for fermionic systems was established in \cite{}, with the sharpest results given in \cite{}.

Our work builds on a completely different approach to propagation bounds originally introduced in \cite{SigSof2} for few-body quantum mechanics in continuous space and further developed in \cite{Skib, HeSk, BoFauSig,APSS,BFLS,BFLOSZ}.
%The approach of  \cite{SigSof, Skib, HeSk, APSS, BoFauSig} is  based on the method of differential inequalities for propagation observables and commutator expansions. It is fairly different from  the existing approaches in the literature on Lieb-Robinson bounds.
%The approach of \cite{SigSof, SigSof2} was extended
The method was recently extended to  Bose-Hubbard Hamiltonians with long-range hopping in \cite{FLS, FLS2} and we draw on the insights developed therein.

%\cjz{The next paragraph moved up from Conclusion and merged with references suggested by the referee}

There has recently been intensive research activity concerning propagation bounds for bosonic lattice systems. Results similar to Theorem \ref{thmMVB}, \ref{thmLCA} and \ref{thmLRB} %and \ref{thm3}
have %previously 
been obtained in \cite{SHOE}, \cite{YL} and \cite{KS2,KVS} for the case of nearest-neighbor hopping. Earlier influential works derived Lieb-Robinson bounds for systems of harmonic oscillators \cite{PhysRevLett.111.230404} which can be coupled to a finite-dimensional quantum system \cite{MR2544074} or perturbed by an anharmonic interaction \cite{MR2472026}; see also \cite{PhysRevLett.115.130401,MR3457926}.

%\ml{Other highly influential works \cite{MR2544074,PhysRevLett.111.230404} consider coupled harmonic oscillator systems coupled to spin degress of freedom and proved propagation estimates which include certain long-range interactions. %\cite{MR2544074} %.considers long-range interactions, see eq.~(35) therein for their result in this case. %\cite{PhysRevLett.115.130401} considers approximation by local observables in models with nearest neighbour interaction, and \cite{MR3457926} considers the error introduced by discretization of continuous model. The error bounds obtained in \cite{} are comparable to our Thm.~\ref{thmLCA}. }

Of particular interest is \cite{KVS} which considers nearest-neighbor hopping and derives superlinear light cones $|x|\sim t \log\,t$ for particle propagation respectively $\abs{x}\sim t^d\,\mathrm{polylog}\,t$ for information propagation, where $d$ is the lattice dimension. Regarding particle propagation, Theorem \ref{thmMVB} Part (ii) extends the first-moment bound in \cite{KVS} from nearest-neighbor interactions to any finite range and it improves the light cone to exactly linear $|x|\sim t$ under a mild initial-density assumption. Considering information propagation, we obtain much stronger linear light cones $|x|\sim t$  (and various information-theoretic consequences) for general long-range interactions under the assumption of initially localized states. The results thus point to a certain dichotomy in the information-propagation capabilities of even long-range lattice bosons depending on the localization of the initial state, an effect first found in \cite{SHOE}.

%
%\paragraph*{Notation.}
%We fix {the underlying lattice $\cL$
%%	, with  grid size $\ge1$, and the domain $\Lam\subset\cL$,
%	 and do not display it in our notations.} We denote by
%$\|\cdot\|$  
%the norm of operators on $\cF$.  All quantities and equations we work with are dimensionless and, in our units, the Planck constant is set to $2\pi$ and speed of light, to one ($\hbar=c = 1$).  

\section{Setup and main results}\label{sec:setup}

%For symmetric $h$ and $v$, the Hamiltonian $H$ in \eqref{H} is symmetric and therefore self-adjoint. 
%To show the latter, one observes that the number operator $N\equiv N_\Lam$, where $N_X=\sum_{x\in X}a_x^*a_x$, commutes with $H$. 
%Since {the operators $H_n=H\upharpoonright_{\Set{N=n}}$ are symmetric and bounded, they are  self-adjoint. Hence} so is $H=\oplus_{n=0}^\infty H_n$ as an infinite direct sum of self-adjoint operators. Therefore the propagator $e^{-itH}$ is well-defined for every $t\in\Rb$. \DETAILS{See \cite[Sect.~5.2.1]{BrRo} for detailed discussions.}

When testing observables against quantum states, we identify the density matrix $\rho$ with the linear functional 
$\om(A)\equiv \om_\rho(A)=\Tr(A \rho)$ on observables $A$. We consider the time evolution of observables in the Heisenberg picture
	\begin{equation}\label{1.5}
		\al_t(A)=e^{itH}Ae^{-itH}\text{, so that }\om_t(A)=\om(\al_t(A)).
	\end{equation}
	%We denote by $\mathcal D$ the domain of $\Ad_H:A\mapsto [A,H]$ and 
	We consider initial quantum states satisfying 
	\begin{equation}\label{g0-cond}
	\om\in\cD ,\quad 	\om(N^2)<\infty.
	\end{equation}
	\noindent where $\mathcal D$ is the domain of the commutator with $H$ \footnote{{Here is the rigorous
definition of $D$. Let $S(\mathcal{F})$ be the Schatten class  of quantum states, i.e.~positive operators $\rho$ on the bosonic Fock space $\mathcal{F}$ with $\mathrm{Tr}(\rho)=1$. 
				We denote by $\mathcal{D}$	the domain of $A\mapsto [A,H]$, given by $\mathcal{D}=\{\rho\in S( \mathcal{F} ) \mid \rho \mathcal{D}(H)\subset\mathcal{D}(H)\text{ and }[H,\rho]\in S(\mathcal{F})\}$,
				and we write $\omega\in\mathcal{D}$ if $\omega=\omega_\rho$ for some $\rho\in\mathcal{D}$}}. Examples include initial states of fixed particle number.

%	For each $\rho\in\cD$, eq.~\eqref{vNL} has a unique solution with initial state $\rho$, given by $
%	\rho_t\equiv \al'_t(\rho)=e^{-itH}\rho e^{itH}$.  This evolution preserves total probability, i.e.~$
%		\Tr(\rho_t)\equiv \Tr(\rho)$, as well as the eigenvalues of $\rho$.
%In this Letter, we present several bounds imposing general constraints on the many-body quantum evolutions.

For any region $X\subset \Lam$, we write $d_X(x)=\inf_{y\in X}\abs{x-y}$ for the associated distance function and $X_\eta=\Set{x:d_X(x)\le \eta}$ for the $\eta$-enlarged set with $\eta>0$. 

 A key role is played by the first moment of the hopping matrix  
	\begin{equation}\label{kappa}
		\kappa=\sup _{x\in\Lam} \sum_{y\in\Lam} \abs{h_{xy}}\abs{x-y}.
	\end{equation}
	which as we will see gives an explicit, calculable bound on the maximal velocity (i.e., the light cone slope).
	%{The number $\kappa$ bounds the norm of the $1$-particle group velocity operator $i[h,x]$,  see \remref{remk} below.}	
	We recall that our standing assumption is that \eqref{k-cond} holds for some $p\geq 1$.

Our first result is a general maximal velocity bound (MVB) for particle transport. Here and below we say that the interaction is of finite range $R>0$ if $
	h_{xy}=0$ for all $x,\,y\in\Lam$ with $|x-y|>R$.

\begin{theorem}[MVB]\label{thmMVB}
	For every $v > \kappa$,  the following holds:

		(i) There exists $C=C(p,\|h\|_p,v )>0$ such that  for all $\eta\ge1$, $X\subset \Lam$, 
		%		we have the following estimate $\forall$ 
	and any initial state $\omega$,  \begin{equation}
			\label{MVE}
	\sup_{	\abs{t}< \eta/v}		 \omega_t(N_{X})
			\le %\big(1+C\eta^{-1}\big) 
			C(\omega(N_{X_\eta}) +\eta^{-p} \omega(N)).
		\end{equation}
	
		(ii) Suppose that the interaction is of finite range $R$. Then there exists $C=C(p,\|h\|_p,v ,R)>0$ such that for any $X\subset \Lam$, there exists $\eta_0=\eta_0(X)\ge1$ such that for any $\eta\ge\eta_0$ and initial state $\omega$ of controlled density (i.e., $0<a\leq \om(n_x)\leq b<\infty$ for all $x$), we have
			 		\begin{align}
	 	\label{TSMVE} &\sup_{	\abs{t}< \eta/v}\om_t(N_{X})
	 	\le %\big(1+C\eta^{-1}\big) 
	 	C	\frac{b^2}{a^2}\om(N_{X_\eta}) 	 \end{align}	
\end{theorem}

%\jz{Part (i) of \thmref{thmMVB}, corresponding to \cite[Thm.~2.1]{SiZh}, is proved in \cite[Sect.~3]{SiZh}.
%Proof of Part (ii) is found in the Supplementary Material.}

%\cjz{Reformulated (ii) according to the proof in \secref{secPfMain}. Does not affect asymptotic behavior for $\eta\gg1$.}
The key steps in the proof of Theorem 1 are given in Section III below with details relegated to \cite{SiZh,SM}.

%\noindent The proof of part (ii) uses a novel analytical machinery and can be found in \cite{SM}. The proof of part (i) as well as the following statements can all be found in \cite{SiZh}.
	%{Continuing with}  terminology of \cite{SigSof2, FLS,FLS2}, %{MR1029392}, 
	%we call such an estimate the \textit{maximal velocity bound} (MVB).
	%We outline the proof of Theorem \ref{thmMVB} in \secref{secShortPfs}.
%Here and below, an operator inequality $A\le B$ means that $\om(A)\le \om(B)$ holds for all states  $\om$ satisfying \eqref{g0-cond} and so the bound holds for all possible states. 

Theorem \ref{thmMVB} part (i) says that even bosons with long-range hopping cannot propagate super-ballistically. The error term $\eta^{-p} N$ grows in the thermodynamic limit and so the propagation is only controlled on macroscopic length scales $\eta\gg N^{-1/p}$. This restriction is removed in part (ii) under additional assumptions.

%This restriction can be removed in various situations. As part (ii) of the theorem shows, if the interaction decays sufficiently rapidly (e.g., subexponentially) such that $p$ can be chosen arbitrarily and we restrict to initial states of bounded positive density (a relatively weak assumption which can be weakened further \cite{LRZ}), then $N$ can be removed. 
	
%This gives for first time a bona fide ballistic particle propagation bound for \textit{any} initial state for a class of infinite-range bosons. (Recall that \cite{KVS} in their landmark paper considered nearest-neighbor hopping.) 
	
The $N$-dependence of the remainder can also be removed if we focus on macroscopic particle transport as follows.
	For a given $S\subset \Lam$, we define  the (macroscopic)  local relative particle numbers as $
		\bar N_S=\frac{N_S}{N_\Lam}. $
	For $0\le \nu\le 1$, we write $P_{\bar N_S\le \nu},\,P_{\bar N_{S^\cp}\ge \nu}$ for spectral projections associated to $\bar N_S$.
	\begin{theorem}[MVB for macroscopic particle transport]\label{thmPT}
		Suppose the initial state $\om$  satisfies  $ \om(P_{\bar N_{X_\eta}\ge \nu})=0$
		with some $\eta\geq 1,\,\nu\ge0,\, X\subset \Lam$. Then, for all $	\nu'>\nu$, $v>\kappa,$ there exists $C=C(p,\kappa_p,v,\nu'-\nu)>0$ such that 
%		we have the following estimate $\forall$ 
		$\abs{t}< \eta/v$:
		\begin{equation}\label{6.2}
			\om_t\del{P_{\bar N_{X}\ge \nu'}}\le C\eta^{-p}.
		\end{equation}
	\end{theorem}
		
	In words, this result \eqref{6.2} says that macroscopic particle clouds of $(\nu'-\nu)N$ particles travel at most with speed $\kappa$. %even for long-range interactions. %Note that estimate \eqref{6.2} holds for rather general initial states and extends to the thermodynamic limit.
	Theorems \ref{thmMVB}--\ref{thmPT} hold without Assumption \eqref{k-cond} on $w$.

%Returning to Theorem \ref{thmMVB}, another way to remove the $N$-dependence of the error term is to consider a different situation where 

Let us now suppose that all the bosons are initially localized in a region $X$, i.e., the initial condition $\om$ satisfies $\om(N_{X^\cp})=0$. Then, for all $\abs{t}<\eta/v$ we can apply Theorem \ref{thmMVB} to complements to obtain
 \begin{align}
		\om_t(N_{X^\cp_\eta})=\om(\al_t(N_{X^\cp_\eta}))\le C\eta^{-p}\om(N_X).
	\end{align} 
	Here we write $X^\cp{=\Lam\setminus X}$ and set
%	and $N_X = \sum_{x\in X}a^*_xa_x.$ We note that $[N_X,N_Y]=0$ $\forall$ $X,\,Y\subset \Lam$, 
	 $X_\eta^\cp\equiv(X_\eta)^\cp$.
%At the last step, we used that under $\om(N_{X^\cp})=0$ one has %\eqref{emptyshellcond}, \begin{equation}\label{NNX}
		%$\om(N^p)=\om(N_X^p),\ p=1,2.$ %\end{equation}
 	%In other words, up to polynomially vanishing probability tails, the particles propagate within the strictly linear light cone (LC) $		X_{vt}\equiv \Set{d_X(x)\le vt}$. %for every fixed $c>\kappa$ and all $t$.  %Put differently, the probability that particles are transported from $X$ to any test (or probe) domain $Y$ outside the LC $X_{ct}$  is of the order $O (\eta^{-n})$, where $\eta=\dist(X,Y)$. 

	It is this setup that we will use for proving the light cone approximation and LRBs next. %We begin with the light cone approximation of the time evolution \eqref{1.5}.
	%The main result of this section and the next one concerns {the evolution of general \textit{local observables.}}
	We say that an operator $A$ acting on $\cF$ is \textit{localized} in $X\subset \Lam$ (in symbols, $\supp A\subset X$) if $
		\sbr{A,a_x^\sharp}=0 $ for all $ x\in X^\cp,$ where $a_x^\#$ stands for either $a_x$ or $a_x^*$.
%
%The support of an generally spreads over the entire space as soon as $t>0$. Nonetheless, 
In \thmref{thmLCA} below, we show that the evolution of initially localized observables under \eqref{1.5} is approximated by a family of observables localized within the LC of the initial support. 
	
For any subset $S\subset \Lam$, we define the \textit{localized evolution} of observables as $
		\al_t^S(A)=e^{itH_S}Ae^{-itH_S},$
	where $H_S$ is defined in  \eqref{H} with $S$ in place of $\Lam$	 and
	\begin{equation}\label{BX}
		\cB_S=\Set{A\in\cB(\cF):[A,N]=0,\, \supp A\subset S}
	\end{equation}
denotes bounded particle number conserving operators localized in $S$.
%	where $\cB(\cF)$ is the space of bounded operators on $\cF$.
%	%Suppose an observable $A$ is initially localized in $X\subset \Lam$. 
%	Then,   one can check that $\forall$ $S\subset \Lam$, $A\in \cB_S$, and $t\in\Rb$, we have $\al_t^S(A)\in\cB_{S}$.
%We have the following result:
	\begin{theorem}[LC approximation of quantum evolution]\label{thmLCA}
		Suppose that the initial state $\om$ satisfies \eqref{g0-cond}  and 
		\begin{equation}\label{emptyshellcond}
			{ \om(N_{X^\cp})=0,}
		\end{equation}
		for some $X\subset \Lam$.
{Then, for every $v>2\kappa$, there exists $C=C(p,\kappa_p,v)>0$ such that for all }$\xi\ge1$ and $A\in \cB_X$,   the full evolution $ \al_t(A)$ is approximated by the local evolution $ \al_t^{X_\xi}(A)$, for all $\abs{t}<\xi/v$, as
		\begin{equation}\label{lcae} %{eq:loc'}
			\begin{aligned}
				\abs{\om{	\del{\al_t(A)-\al_t^{X_\xi}(A)} }} \le C\abs{t}\xi^{-p}\Norm{A}\om(N_X^2).
			\end{aligned}
		\end{equation}
		
	\end{theorem} 
%	\jz{To find the relevant number like the distance over which one would like to communicate, we consider a state $\om$ with $30$  particles uniformly distributed over a set $X$ with $30$ sites, i.e. $\rho\equiv 1$ on $X$ and zero elsewhere. Then $\om(N_X^2)=30$. \textbf{(???)} We take $\xi$ to be the size of a typical superconducting quantum computer, say $1$m (in diameter). Thus, in  our length scale (atomic units): 	  $\xi \approx 2\times 10^{11}$ \textbf{(???)} in \eqref{lcae} and $t< \xi/c \approx ???$} 
%	\cjz{For this we need a numerical value for $c$, say $c=3\kappa$ with $\kappa$ in \eqref{kappa}. What is the order of magnitude of   $\kappa?$  
%		
%		Objective (???): determine a number for 
%		\begin{align}\label{object}
%			C\abs{t}\xi^{-n}\om(N_X^2)
%		\end{align}
%		with $C=1$, $t=\xi/c$, $c=3\kappa$. 
%	But $C$ presumably depends on $c$ and $\kappa$ in a nontrivial way. There could be many factors of $(c-\kappa)^{-1}$ or $c^{-1}$ hidden in $C$. So I do not know if setting $C=1$ in \eqref{object} gives a meaningful approximation for its value (could be very many orders of magnitude off). }
	
%We now come to the weak Lieb-Robinson bound.
%	 leads to a Lieb-Robinson-type bound: 
%	for general interacting quantum many-body systems:

	We remark that the term $\omega(N_X^2)$ which appears here and in the following is controlled by the state's local energy density and particle density near $X$ \footnote{{Fix a region $X$. We can apply Cauchy-Schwarz twice (once to the hopping term and once to $n_xn_y\leq \frac{n_x^2+n_y^2}{2}$) to obtain the operator inequality			$N_X^2 \leq C|X| (H_X+N_{X_1})$. This shows that for finite regions $X$ and for states of locally bounded energy and particle density, the expectation of $N_X^2$ is well-controlled.}}.

	\begin{theorem}[Weak Lieb-Robinson bound] 
		\label{thmLRB}
		Suppose the assumptions of \thmref{thmLCA} hold with $n\ge1,\,X\subset\Lam$. Then, for every $v>2\kappa$, there exists $C=C(p,\kappa_p,v)>0$ such that for all $\xi\ge1$,  $Y\subset\Lam$ with $\dist(X, Y)\ge2\xi$, and  operators $A\in \cB_X,\,B\in \cB_Y$, 
		we have, for all $\abs{t}< \xi/v$:
		\begin{align}
			\abs{\om\del{[ \al_t(A), B] }}  
			\le C\abs{t}\Norm{A}\Norm{B}\xi^{-p}\om(N_X^2).\label{LRB}
		\end{align}
	\end{theorem} 
		%\jz{Thms.~\ref{thmLCA}--\ref{thmLRB}, corresponding to \cite[Thms.~2.2-3]{SiZh}, are proved in .}

	%\thmref{thm3} is proved in \secref{sec:6}. 
	We call a bound  of the form \eqref{LRB} the \textit{weak Lieb-Robinson bound (LRB)} because unlike the classical LRB, estimate \eqref{LRB} depends on a subclass of  states.

	%It also only provides power-law, rather than exponential, decay outside the light cone but the latter is a well-known phenomenon for long-range interactions \cite{}.

	\section{Applications}\label{secAppl}

%	Estimate \eqref{LRB} shows that, with the probability approaching $1$ as $t\ra \infty$, an evolving family of observable $A_t=\al_t(A)$ remains commuting with any other observable supported outside the LC \eqref{2.8a} $\forall$ {$c>2\kappa$,} provided the supports of these observables are separated  {by initially empty regions}. \DETAILS{This implies that the maximal speed of quantum propagation for \eqref{1.5} is bounded  (up to an absolute constant) by the number $\kappa$ defined in \eqref{kappa}.} 

	\paragraph{Propagation/creation of correlations.}
	
%	In this section we address the following questions (c.f.~\cite{BHV,NachOgS}):  
	Assuming a state $\om$ is weakly correlated in a domain $Z^\cp\subset \Lam$, how long does it take to create substantial correlations in $Z^\cp$? 	For subsets $X,Y,Z\subset \Lam$,  let $d_{XY}=\dist(X,Y)$ and  $
	{		d_{XY}^{Z^c}=\min (d_{XY}, d_{XZ^c}, d_{YZ^c}).}$
%			To address this question, we first define what we mean by 
			%The notion of 
			%weakly correlated state is defined as follows:  

	\begin{definition}[Weakly correlated state]\label{def:WC'}
	
	%\item an operator (observable) $A$ is localized in $U\subset \Lam$ if $A$ is of the form  $A=A_U\oplus \one$, where the operator $A_U$ acts on $\cF(L^2(U))$, w.r. to the decomposition   \[\cF(L^2(\Lam))=\cF(L^2(U))\otimes \cF(L^2(U^c)).\] 

	We say a state $\om$ is weakly correlated  inside a subset $Z$ on the length scale $\l>0$ if there are $C>0$, $p\ge1$ such that for all $X,\,Y\subset Z$ with $d_{XY}^{Z^c}>0$ and operators $A\in\cB_X,\,B\in\cB_Y$ (see \eqref{BX}), we have	\begin{equation}%\label{221}
		|\om^c(A,B)|\le C (d_{XY}^{Z^c}/\l)^{1-p} \norm{A}\norm{B},
	\end{equation}
	%with $C>0$,
	where $\om^c(A,B)=\om(AB)-\om(A)\om(B)$. In this case, we write $\WC(Z, \l,C,p)$. %Similarly, we define WC-$d$ in a subset $Z\subset \Lam$ by requiring $X,\,Y\subset Z$ in the preceding definition and replacing $d_{XY}$ by the distance $d_{XY}^Z$ in $\cL$ in the Riemannian metric vanishing on $Z$. 
	
	%	\DETAILS{If we shift $n\to n-1$, then every estimate that carries a power of $n$ will be changed to $n+1$. This is a very tedious process and will make it harder  for the readers to trace parameters in equations.}
	
\end{definition}

%As for exponentially decaying correlations,   
%\jz{\sout{Thus, $\l$ provides a bound on the correlation length.}}
%The following result shows that the maximal speed for the propagation/creation of correlations  is bounded by $3\kappa$:
\begin{theorem}[Propagation/creation of correlation]\label{thm5}
	%Suppose the assumptions of \thmref{thm2} hold  for some  constant $\eta\ge 1$.
		Let  $Z \subset \Lam$ and suppose the initial state $\om$ satisfies \eqref{g0-cond},   $
	\om(N_{Z})=0,$ and is $\WC(Z, \l,C,p)$. 
	
	Then  $\om_t$ is  $\WC(Z, 3\l, C\om(N_{Z^c}^2),p)$
	for  all $ \abs{t}<\l/3\kappa$.
	
%	specifically, for all  $A\in\cB_X,\,B\in\cB_Y$	supported in $X,\,Y\subset Z$ with $d_{XY}^{Z^c}>0$ and $\abs{t}<\l/3\kappa$,
%	\begin{equation}\label{221}
%		|\om^c_t(AB)|\le C\om(N_{Z^c}^2) (d_{XY}^{Z^c}/3\l)^{1-n}\norm{A}\norm{B}.
%	\end{equation}	  %For short-range (i.e.~exponentially decaying) interactions, \eqref{221} holds %with $\l$ replaced by $3\l$ and
%	for all $n\ge1$. 
\end{theorem}
		
		%\jz{Note that estimate (14) is uniform for all $\lambda$.}
%For the second statement, we note that for short-range interactions, condition \eqref{k-cond} is valid $\forall$ $n\ge1$. 

	\paragraph{Constraint on the propagation of quantum signals.}
	
	The weak LRB  \eqref{LRB} imposes a direct constraint on the speed of quantum messaging (cf.~\cite{BHV, FLS2, Pou}).  Assume that Bob at a location $Y$ is in possession of a state $\rho$ and an observable $B$ and would like to send a signal through the quantum channel $\alpha_{t}'$ to Alice who is at $X$ and who possesses the same state $\rho$ and an observable $A$. {To send a message}, Bob 
%	uses $B$ as a Hamiltonian to evolve $\rho$ for a time   $r>0$, and then 
	sends Alice the state $\rho_r=e^{-iBr}\rho e^{iBr}$. %by the quantum channel,
%as $\alpha_{t}'(\rho_r)$. %To send the message ``$0$'', Alice simply sends $\alpha_{t}'(\rho)$. 
To see whether Bob sent his message, %``$0$'' or ``$1$'', 
Alice computes the difference between the expectations of $A$ in the states $\al_{t}'(\rho_r)$ and $\al_{t}'(\rho)$, which we call the \textit{signal detector}: $\SD(t,r)=\tr\big[A\al_{t}'(\rho_r)-A\al_{t}'(\rho) \big]$.	
\DETAILS{	Assume that Alice at a location $X$ is in possession of a state $\rho$ and an observable $A$ and would like to send a signal through the quantum channel $\alpha_{t}'$ to Bob who is at $Y$ and who possesses the same state $\rho$ and an observable $B$. {To send a message}, Alice uses $A$ to evolve $\rho$ for a time, $r>0$, as $\rho_r=e^{-iAr}\rho e^{iAr}$, and then sends Bob the resulting state   $\alpha_{t}'(\rho_r)$. %To send the message ``$0$'', Alice simply sends $\alpha_{t}'(\rho)$. 
	To see whether Alice sent her message, %``$0$'' or ``$1$'', 
	Bob computes the difference between the expectations of $B$ in the states $\al_{t}'(\rho_r)$ and $\al_{t}'(\rho)$, which we call the \textit{signal detector}, given by %Using the approximation $\rho_r=e^{-iAr}\rho e^{iAr}\approx \rho - r i[A, \rho]$, this gives
$ \SD(t,r)%=\tr\big[B\al_{t}'(\rho_r)-B\al_{t}'(\rho) \big]&
		%\approx  r\tr\big[B\al_{t}'( i[\rho, A])\big]&= r\tr\big( i[A,\al_{t}'(B)] \rho\big).
		=\Tr\sbr{B\al_t'(\rho_r)-B\al_t'(\rho)} %\notag\\&= r\tr\big( i[A,\al_{t}'(B)] \rho\big). 
$.}
	%The weak LRB  \eqref{LRB} implies:
{\begin{theorem}[Bound on messaging time]\label{thm6}
	Under the assumptions of \thmref{thmLCA}, for every $v>4\kappa$, there exists $C=C(p,\kappa_p,v)>0$ such that for all  $\xi\ge2$,  $X,\,Y\subset\Lam$  with {$\dist(X, Y)\ge2\xi$,} and operators  $A\in\cB_X$, $B\in \cB_Y$ with  $\norm{B}_n<\infty$  (see after \eqref{k-cond}),
	we have, for all $r,\abs{t}< \xi/v$:
	\begin{equation}\label{7.2}
		\abs{	\SD(t,r)}\le Cr\abs{t}\xi^{-p}\Norm{A}\Norm{B}\Tr(N_X^2\rho).	\end{equation}
\end{theorem}}

	\paragraph{Bound on quantum state control.}
	{For any $S\subset \Lam$, we denote by  $\cF_{S}$  the bosonic Fock space $\cF_S=\oplus_{n=0}^\infty \otimes^n_\mathrm{S} \ell^2(S),$ where $\otimes_\mathrm{S}$ stands for the symmetric tensor product. Due to the tensorial structure $\cF_\Lambda \simeq \cF_Y\otimes_\mathrm{S} \cF_{Y^\cp}$ (see \cite[App.~A]{SiZh}), we can define  the partial  trace $\mathrm{Tr}_{\cF_{Y^c}}$ over ${\cF_{Y^c}}$. We define the {\it restriction} of a state $\rho$ to the density operators on the local Fock  space $\cF_{Y}$, $Y\subset \Lam,$ %is defined as
		by $[\rho]_Y= \mathrm{Tr}_{\cF_{Y^c}} \rho$. %, where %$\mathrm{Tr}_{Y^c}:$ (states on $\Lam)  \ra$ (states on $Y$)
		%	$\mathrm{Tr}_{Y^c} \gamma$ is the {\it trace over ${\cF_{Y^c}}$} %\footnote{} $\mathrm{Tr}_{Y^c}$ is a partial trace 
		%defined e.g. by requiring $\mathrm{Tr}_{\cF_{Y}}(A\mathrm{Tr}_{\cF_{Y^c}}\gamma)=\mathrm{Tr}((A\otimes \one_{\cF_{Y^c}})\gamma)$ for every $A$ acting on $\cF_{Y^\cp}$.
	}
	
	Let  $\tau$ be a quantum map (or {\it state control
		map}) %unitary operator $A$  on a region
	supported in $X$. % and denote 	
	Given  a density operator $\rho$, our task is to design $\tau$ so that  at some time $t$, the evolution $\rho^\tau_t=\al_t(\rho^\tau)$  of  the density operator $\rho^\tau= \tau(\rho)$  has the restriction $[\rho^\tau_t]_{Y}$ to $S(\cF_{Y})$, which is close to a desired state, say~$\s$. 	To measure the success of  the  transfer operation, one can maximize the figure of merit  $
		F([\rho^\tau_t]_Y, \s),$
	where $F(\rho,\sigma)=\|\sqrt{\rho}\sqrt{\sigma}\|_{1}$ is the fidelity. 
	%\medskip
	%Here $\|\rho\|_{S_1}=\Tr(\abs\rho)$ is the Schatten $1$-norm.
%	\footnote{Note that $0\le F(\rho,\sigma)\le 1$, with $F(\rho,\sigma)=1$ if and only if $\rho=\sigma$. If $\s=\abs{\phi\rangle\langle\phi}$, then $F(\rho,\si)=\sqrt{\br{\phi,\rho\phi}}$ and therefore$F(\rho,\sigma)=0$ if  $\rho\phi=0$.}.
	%\sout{One would like to find $\tau$ that maximizes the figure of merit $F([\rho^\tau_t]_Y, \s)$.}%This leads
	%Using this figure of merit, one can estimate the upper bound on the state transfer time.
	
	To show that the state  transfer is impossible in a given time interval, we compare $\rho^\tau_t$ and $\rho_t=\alpha_t(\rho)$ by using   $
		F([\rho^\tau_t]_Y, [\rho_t]_Y),$
	as a figure of merit (cf.~\cite{EpWh,FLS2}), and try to show that it is close to $1$ for $t\le t_*$ and for all state preparation (unitary) maps  $\tau$   localized in $X$. If this is true, then using $\tau$'s   localized in $X$ does not affect states in $Y$.
	Let $\tau(\rho)=U\rho U^*\equiv \rho^U$, where $U$ is a unitary.
	
	\begin{theorem}[Quantum control bound]\label{thm:qst}

		Let $\om$ be a pure state. Under the assumptions of \thmref{thmLCA}, for every $v>8\kappa$, there exists $C=C(p,\kappa_p,v)>0$ such that for all  $\xi\ge4$, $Y\subset \Lam$ with $\dist(X,Y)\ge2\xi$, and unitary operator	 $U\in\cB_X$ (see \eqref{BX}), we have, for all $\abs{t}< \xi/v$:
		\begin{equation*}
			\begin{aligned}
				F([\al'_{t}(\rho)]_Y,[ \al'_{t}(\rho^U)]_Y) \geq& 1- C\abs{t}\xi^{-p}\Tr(N_X^2\rho).
			\end{aligned}
		\end{equation*}

	\end{theorem}

The estimate above imposes a lower bound on the time for the best-possible quantum control protocols for  quantum many-body dynamics and bounds quantum state transfer as in \cite{EpWh,FLS2}.

\paragraph{Spectral gap and decay of correlation.}

%Hastings realized that the Fourier transform connects LRBs to ground state correlation properties. Here we give a bosonic analog. %dynamical LRBs can be effectively combined with the Fourier transform to control ground state correlation properties. We also obtain such a result in the bosonic setting.
%In this section, we use  the weak LRB \eqref{LRB} to derive a bound for the ground state correlation functions

%Denote by $\Om$ the ground state of the Hamiltonian $H$ in \eqref{H}\DETAILS{ (i.e.~$H\Om=0,\,\norm{\Om}_\cF=1$)}. 
%We have the following result:
\begin{theorem}[Gap implies decay of correlations]
	\label{thm:gap}
	Suppose $H$ in \eqref{H} has a spectral gap of size $\g>0$ at the ground state energy. 
	Suppose the assumptions of \thmref{thmLCA} hold with $n\ge1,\,X\subset\Lam$,
	{and $\om$ is the ground state.} Then, there exists $C=C(n,\kappa_n)>0$ such that for all  $\xi\ge1$,  $Y\subset\Lam$ with $\dist(X, Y)\ge2\xi$, and  operators $A\in \cB_X,\,B\in \cB_Y$, we have:
	\begin{equation}\label{281}
		\abs{\om({BA})}\le C\norm{A}\norm{B}(\g^{-1}\xi^{-2}+\xi^{1-p}\om(N_X^2)).
	\end{equation}
\end{theorem}

	%\DETAILS{Pls check my transcription in the last paragraph }

%\paragraph{Extensions.}
%
%Results from the preceding subsections can be extended to (a) time-dependent one-particle and two-particle operators $h$ and $v$ satisfying \eqref{k-cond} uniformly in time, (b) few-pody potentials in \eqref{H}, (c)  observables which are polynomials in $\Set{a_x,\,a_x^*}_{x\in\Lam}$, and (d) fermi systems.
%%, (e) continuous systems obeying suitable modifications of \eqref{k-cond}. 

\section{Key steps in the proofs}\label{secShortPfs}
	
We sketch the key ideas involved in the proofs and refer to \cite{SM} and \cite{SiZh} for the full details. The method is an adaptation of the ASTLO (adiabatic spacetime localization observables) approach developed in \cite{FLS,FLS2}. %In this approach, one dynamically tracks the particles outside of the light cone through recursive monotonicity estimates leveraging the adiabatic parameter called $s$ below. %The main novelties compared to prior works are (i) the treatment of infinite-range $2$-particle interactions and (ii) the geometrical universality of the bound which does not depend on specific region shapes etc.

\paragraph{\thmref{thmMVB}.}
%Recall that the second quantization $\dG$ of $1$-particle operators $b$ on the $1$-particel Hilbert space $ \ell^2(\Lambda)$ 
%is given by  
%%\begin{equation}\label{3.1s}
%	$\dG(b)=\sum_{\Lambda\times\Lam} b_{xy} a_x^* a_y,$
%%\end{equation}
%where $b_{xy}$ is the matrix  of $b$. 
For a function $f:\Lambda\to\Cb$, we define its second quantization
$\hat f= \dG(f)=\sum_{x\in\Lambda}f(x)a_x^*a_x. %, \qquad N_S=\hat\chi_S
$
%for the second quantization $\dG(f)$ of the multiplication operator defined by $f$. 
%We denote by $\chi_S^\sharp$ the characteristic function of a subset $S\subset\Lambda$. %For  $f=\chi_S^\sharp$, the above gives the local particle number operators  $N_S\equiv \dG(\chi_S^\sharp) $. 
%For a differentiable real function $f$, we write $\hat f'\equiv \dG(f')$ and $\hat f_{ts}'\equiv  \dG(f_{ts}')$, where $f_{ts}'\equiv f'\del{\tfrac{d_X-vt}{s}}$.
As in \cite{FLS, FLS2}, we control the time evolution  associated to \eqref{H} by recursive monotonicity estimates for the \textit{ASTLOs}:
\begin{equation}\label{chi-ts}
\hat\chi_{t s}=\dG(\chi_{t s}),\quad 	\chi_{ts}=\chi\del{\tfrac{d_X-v't}{s}}, 
\end{equation}        
where $ s> t \ge0$, $d_X$ is the distance function to $X$, $v'=\frac{v+\kappa}{2}$ and $\chi$ and $\chi:\mathbb R\to[0,1]$ belongs to the set $\cX$ of smooth monotonic cutoff functions which interpolate between $0$ and $1$ and satisfy $\sqrt\chi’\in C^\infty$ and $\mathrm{supp}\, \chi’\subset (0,v-v’)$.
%\DETAILS{where $\delta=c-v$ with $c$ and $v$ given in the statement of \thmref{thm1} and \eqref{chi-ts}, respectively.  We note that $\chi\ge0$ for each $\chi\in\cX$.  }

For a differentiable path of observables, define the Heisenberg derivative % of $A(t)$,  as
$D A(t)=\frac{\partial}{\partial t}A(t)+i[H, A(t)],$ with %the property
\begin{align}	\label{dt-Heis}
%D A(t)=\frac{\partial}{\partial t}A(t)+i[H, A(t)],\ \text{ s. t. }\	
&\di_t\al_t(A(t)) =\al_t(DA(t)). 
\end{align}
%We use \eqref{dt-Heis} to prove a key differential inequality:

\begin{proposition}[RME]\label{prop:RME} %{timederivmaintxt}
	Suppose the assumptions of \thmref{thmMVB} hold.
	Then, for every $\chi\in\cX$, there exist $C=C(p,\kappa_p,\chi)>0$ and functions $\xi^k=\xi^k(\chi)\in\cX,\,k=2,\ldots, p$, % (dropped if $n=1$), 
	 %each supported in $\supp \chi$,  
	 such that for all $s,t>0$, we have the operator inequality
	\begin{align}\label{RMB}
		%\begin{aligned}
		D\hat\chi_{ts} 
		\leq &-\frac{(v' -\kappa)\astlo'_{ts}}{s} +C\sum_{k=2}^{p}\frac{\widehat{(\xi^k)'}_{ts}}{s^k}+ \frac{CN}{s^{p+1}}.
		%\end{aligned}
	\end{align}
%	(The sum in the r.h.s. is dropped if $n=1$.)
\end{proposition}
%At the end of this section, \thmref{thm:RME} is reduced to a $1$-particle estimate. 
Since the second  term on the r.h.s. is of the same form as the leading, negative term, estimate \eqref{RMB} can be bootstrapped to obtain an integral inequality with $O(s^{-p})$ remainder which gives part  \thmref{thmMVB}  (i). For the details see \cite[Sect.~3]{SiZh}. 
The proof of  \thmref{thmMVB} part (ii) uses novel techniques (which we call band-limited ASTLOs) and is fully contained in \cite{SM}.
The proof of \thmref{thmPT} is similar to the proof of \thmref{thmMVB} part (i) except that we introduce a ``second-order'' ASTLO by using spectral calculus to approximate the spectral projector as in \cite{FLS}. See \cite[Section 4]{SiZh} for details.

%Namely, we consider another smooth cutoff function $f$ with derivative supported in $(0,\nu'-\nu)$ and define the ASTLOs  $f_{ts}=f(\bar\chi_{ts})$,  for  $\abs{t}<s$.
%Using \eqref{RMB} and the `integral chain rule' $Df_{ts}=\int R_{ts}(z)D \bar\chi_{ts} R_{ts}(z)\,d \tilde f(z),$ where $R_{ts}(z)=(z-\bar\chi_{ts})^{-1}$ for $\Im z\ne0$ and  $d\tilde f(z)$ is a complex measure vanishing for $\Im z=0$, we obtain the RME for $f_{ts}$ and iterate as before.

\paragraph{\thmref{thmLCA}.}

%Denote by $H_Y,\,Y\subset \Lam$ the Hamiltonian defined by \eqref{1.1}  with $Y$ in place of $\Lam$, and let $A_s^\xi\equiv \al_s^{X_\xi}(A)=e^{itH_{X_\xi}}Ae^{-itH_{X_\xi}}$.
%It is easy to check that $\supp A\subset X\implies \supp A_s^\xi\subset X_\xi$  $\forall$ $s\in\Rb,\,\xi\ge0$. 
Let $A_t=\al_t(A)$ and $A_t^\xi\equiv \al_t^{X_\xi}(A)$.
By the fundamental theorem of calculus, we have $	A_t-A_t^\xi=\int_0^t \di_r\al_r(\al_{t-r}^{X_\xi}(A))\,dr.$	 Using identity \eqref{dt-Heis} for $\al_r$ and $\al_{t-r}^{X_\xi}$, we find 
\begin{align}\label{At-differ}A_t-A_t^\xi=i\int_0^t\al_r([R',A_{t-r}^\xi])\,dr, \end{align}
where $R'=H-H_{X_\xi}$. Since $A_s^\xi$ is localized in $X_\xi$, only terms in $R'$ which connect $X_\xi$ and $X_\xi^\cp$ contribute to $[R',A_{t-r}^\xi]$ (see \figref{fig:splitting}). Then we apply the  MVB \eqref{MVE} to estimate $[R',A_{t-r}^\xi]$ similar to \cite{FLS2}. We defer further details to \cite[Sects.~5-6, Appd.~D]{SiZh}.

As mentioned before, Theorems \ref{thm5}--\ref{thm:gap} are by now relatively standard consequences of Lieb-Robinson bounds with some modifications required due to the bosonic nature of the Hamiltonian. In \cite{SiZh}, these results are stated as Theorems 2.4-2.7 and complete proofs are given in \cite[Sections 7-10]{SiZh} respectively.

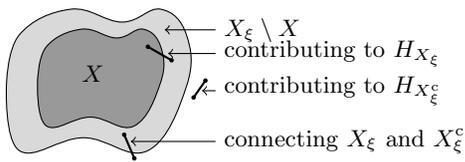
\begin{figure}[t]
	\centering
	\begin{tikzpicture}[scale=.45]
%		\draw  plot[scale=.5,smooth, tension=.7] coordinates {(-3,0.5) (-2.5,2.5) (-.5,3.5) (1.5,3) (3,3.5) (4,2.5) (4,0.5) (2,-2) (0,-1.5) (-2.5,-2) (-3,0.5)};
		
		% 		\draw  plot[shift={(-0.2,-0.25)}, scale=.8,smooth, tension=.7] coordinates {(-3,0.5) (-2.5,2.5) (-.5,3.5) (1.5,3) (3,3.5) (4,2.5) (4,0.5) (2,-2) (0,-1.5) (-2.5,-2) (-3,0.5)};
		
%				\draw  plot[shift={(0.05,0.18)},scale=.32,smooth, tension=.7] coordinates {(-3,0.5) (-2.5,2.5) (-.5,3.5) (1.5,3) (3,3.5) (4,2.5) (4,0.5) (2.5,-2) (0,-1.5) (-2.5,-2) (-3,0.5)};
				
%				\draw  plot[shift={(-0.2,-0.25)}, scale=.76,smooth, tension=.7] coordinates {(-3,0.5) (-2.5,2.5) (-.5,3.5) (1.5,3) (3,3.5) (4,2.5) (4,0.5) (2.5,-2) (0,-1.5) (-2.5,-2) (-3,0.5)};
				
%				\draw [->] (6,.5)--(1.3,.5);
%				\node [right] at (6,.5) {$X$};
%				
%				\draw [->] (6,0)--(2.3,0);
%				\node [right] at (6,0) {$X_\xi$};

						\draw[fill=gray!30] plot[shift={(-0.2,-0.25)}, scale=.76,smooth, tension=.7] coordinates {(-3,0.5) (-2.5,2.5) (-.5,3.5) (1.5,3) (3,3.5) (4,2.5) (4,0.5) (2.5,-2) (0,-1.5) (-2.5,-2) (-3,0.5)};
						
				\draw[fill=gray!80]  plot[scale=.52,smooth, tension=.7] coordinates {(-3,0.5) (-2.5,2.5) (-.5,3.5) (1.5,3) (3,3.5) (4,2.5) (4,0.5) (2.5,-2) (0,-1.5) (-2.5,-2) (-3,0.5)};

				\node [right] at (-.5,.5) {$X$};
				
				\draw [->] (3.7,1.8)--(2.2,1.8);
		\node [right] at (3.7,1.8) {$X_\xi\setminus X$};

		\draw [->] (3.7,1.1)--(2.3,1.1);
		\node [right] at (3.7,1.1) {$\text{contributing to }H_{X_\xi}$};
		
		\draw [thick] (1.7,1.3) -- (2.4,.9);
		\draw [fill] (1.7,1.3) circle [radius=0.05];
		\draw [fill] (2.4,.9) circle [radius=0.05];
		
		\draw [->] (3.7,0)--(3.3,0);
		\node [right] at (3.7,0) {$\text{contributing to }H_{X_\xi^\cp}$};
		
		\draw [thick] (3.05,.-.2) -- (3.4,.3);
		\draw [fill] (3.05,.-.2) circle [radius=0.05];
		\draw [fill] (3.4,.3) circle [radius=0.05];
		
		\draw [->] (3.7,-1.5)--(1.2,-1.5);
		\node [right] at (3.7,-1.5) {$\text{connecting ${X_\xi}$ and ${X_\xi^\cp}$}$};
		
		\draw [thick] (1,-1.3) -- (1.3,-2);
		\draw [fill] (1,-1.3) circle [radius=0.05];
		\draw [fill] ((1.3,-2) circle [radius=0.05];

	\end{tikzpicture}
	\caption{Set $X_\xi$ and splitting of $H$.}
	\label{fig:splitting}
\end{figure} 
%\begin{figure}[t]
%	\includegraphics[scale=.5]{figure1.png}% Here is how to import EPS art
%	\caption{Set $X_\xi$ and splitting of $H$.}
%\label{fig:splitting}
%\end{figure}
%The argument presented above can be adapted to infinite-range $h$ and $v$ satisfying \eqref{k-cond}.

%For infinite-range $h$ and $v$, we refine the argument presented above. Let $X_{a,b}=X_b\setminus X_a$ for $b>a\ge0$. To estimate $[R',A_{t-r}^\xi]$, we \DETAILS{pick a parameter $0< \g\le 1$ and} split the annulus $X_{0,2\xi}$ into four annuli, say $ X_j=X_{j\xi/2, (j+1)\xi/2},\,j=0,\ldots, 3$. In $X_2$, $X_3$, we use the MVB from Thm.~\ref{thmMVB} and in $X_1$, $X_4$, the decay properties of $h_{xy}$ and $v_{xy}$ as $\abs{x-y}\to\infty$. 

%\paragraph{.}Write $A_t=A_t^\xi+\Rem(A,t)$ with $\Rem(A,t)$ satisfying \eqref{lcae}. %{eq:loc'}. 
 %Plugging this into $\om([A_t,B])$ and using that $[A_t^\xi,B]=0$, we obtain \eqref{LRB}.

%Paragraph (a) to obtain \eqref{6.2}.

\section{Conclusions}

We have presented the existence of a linear light cone for many-boson system with long-range hopping and long-range interactions. The results provide practical constraints on sending quantum information in such systems. Our results identify good classes of initial states in which even bosons with long-range hopping and long-range interactions cannot propagate information super-ballistically. Our results thus complement the existence of fast-transfer protocols for nearest-neighbor Bose-Hubbard Hamiltonians \cite{KVS} and quantum spin systems with long-range interactions \cite{KS1,KS2,EldredgeEtAl,TranEtal4,TranEtal5}. %Additional discussion and detailed proofs of the results presented in this Letter are given in \cite{SiZh} and \cite{}. 

%exact linear light cones for long-range models and localized initial states that we obtain in Theorems~\ref{thmMVB}--\ref{thmLRB},  %\ref{thm4},  \ref{thm6}, and \ref{thm:qst} are most closely related to results in \cite{FLS, FLS2} for on-site Bose-Hubbard type interactions and our overall method is similar to these works.
  
 %An important open question is to derive a thermodynamically stable MVB for long-range Hamiltonians of the form \eqref{} in which the error term does not depend on the particle number. This would provide the key input for establishing an LRB for long-range interacting bosons in \textit{general bounded-density initial states} by adapting the techniques of \cite{KVS} using ideas from the analysis of LRBs in long-range spin systems \cite{}.

The results can  be extended to time-dependent and few-body interactions, quantum spin systems and lattice fermions. Extensions to open systems of lattice bosons with applications to estimating the decoherence and thermalization times are also in reach using ideas from \cite{BFLS,BFLOSZ}.

\begin{acknowledgments}
The authors are grateful to J\'er\'emy Faupin and Avy Soffer for enjoyable and fruitful collaborations. The research of M.L.\ and C.R.~ is supported by the Deutsche Forschungsgemeinschaft (DFG, German Research Foundation) through grant TRR 352--470903074. The research of I.M.S.\ is supported in part by NSERC Grant NA7901.
J.Z.\ is supported by DNRF Grant CPH-GEOTOP-DNRF151, DAHES Fellowship Grant 2076-00006B, DFF Grant 7027-00110B, and the Carlsberg Foundation Grant CF21-0680. His research was also supported in part by NSERC Grant NA7901.  
Parts of this work were done while J.Z.\ was visiting MIT. 
We thank the referees for their helpful comments.
%\jz{We thank the referee for pointing out Refs.~\cite{MR2544074,PhysRevLett.111.230404,PhysRevLett.115.130401,MR3457926}.}
\end{acknowledgments}

\bibliography{MVBshortBibfile}

 \widetext
 \pagebreak
 \begin{center}
 	\textbf{\large Supplemental Material for\\ \textit{Information propagation in long-range quantum many-body systems}}
	
	\vspace{.35cm}
	{M.~Lemm, C.~Rubiliani, and J.~Zhang}
		\vspace{.25cm}
		
 \end{center}
 %%%%%%%%%% Merge with supplemental materials %%%%%%%%%%
 %%%%%%%%%% Prefix a "S" to all equations, figures, tables and reset the counter %%%%%%%%%%

\stepcounter{myequation}
%\stepcounter{mytheorem}
 \setcounter{figure}{0}
 \setcounter{table}{0}
 \setcounter{section}{0}
  \setcounter{theorem}{0}
 \makeatletter
 \renewcommand{\theequation}{S\arabic{equation}}
  \renewcommand{\thetheorem}{S\arabic{theorem}}
 \renewcommand{\thefigure}{S\arabic{figure}}
 
 %change counter depth
 %\numberwithin{equation}{section}
% \numberwithin{theorem}{section}
 %\numberwithin{definition}{section}
%  \begin{center}
% 	\uppercase{\textbf{ Thermodynamically Stable Propagation Bound for Finite Range Interactions}}
% \end{center}

%\section{Thermodynamically Stable Propagation Bound for Finite Range Interactions}
%\cjz{8.31 
%	
%	Fixed labeling issues.
%	
%	To do: 
%	
%	-drop $(\cdot)$ in propositional references
%	
%	-match up the proof to the new statement of  Part (ii) of Theorem \ref{thmMVB}
%}
%We want to show a bound on the particle propagation velocity, and to do this we need to specify the two main assumptions. The first is that the interactions are finite range, so that there exists a positive constant $R$ such that if the distance between two points on the lattice is greater than $R$, they will not interact. The second assumption is that in the initial state the density of bosons on the lattice is finite and constant. This means that there is a positive constant $\delta$ such that for every $x\in\Lambda$, $\langle n_x\rangle_0=\delta$. Under these assumptions we can prove the desired bound.

In this Supplementary Material we prove Part (ii) of Theorem \ref{thmMVB}, which is the first thermodynamically stable particle number propagation bound for arbitrary finite-range hopping of Bose-Hubbard type.

The main idea is to take the effective choice of $p$ to depend on $N$ in order to absorb the term $\omega_t(N)$ in the $\eta^{-p}$ in Part (i) of Theorem \ref{thmMVB}. This simple idea turns out to be quite subtle because to avoid blow-up of the other constants in $p$, we have to work with band-limited analogs of cutoff functions. (We say a function is band-limited if its Fourier transform is compactly supported.) We build the ASTLOs using these band-limited functions which ensures uniform control on all derivatives. The price is that, by the uncertainty principle, these functions are no longer compactly supported. Resolving the ensuing technical complications by carefully tracking the localization properties is the content of this supplement. 

Recall that we say the interaction is of finite range $R>0$ if the hopping matrix satisfies $h_{xy}=0$   whenever $|x-y|>R.$. Our goal is to prove the following result. 

 \begin{theorem}[Part (ii) of Theorem \ref{thmMVB}]\label{CT main}
Suppose that the interaction is of finite range $R$. For every $v > \kappa$ 	there exists $C=C(R,v,d)>0$ such that the following holds. For any $X\subset \Lam$, there exists $\eta_0=\eta_0(X)\ge1$ such that for any $\eta\ge\eta_0$ and initial state $\omega$ of controlled density (i.e., $0<a\leq \om(n_x)\leq b<\infty$ for all $x$), we have
\begin{align}
	 &\sup_{	\abs{t}< \eta/v}\om_t(N_{X}) \le C\frac{b^2}{a^2}\om(N_{X_\eta})
  \end{align}	
 \end{theorem}

We introduce the notation 
 \[
 \omega_t(A)=\langle A\rangle_t.
 \]
  For simplicity, we identify $\delta$ such that
 \[
\delta^{-1}\leq  \om(n_x)\leq \delta
 \]
and we work with the single parameter $\delta$ instead of $a$ and $b$ in the following. We remark that the constant $\frac{b^2}{a^2}$ can be replaced by $\frac{b}{a}$ if $X$ is assumed to be a ball.\\

\textbf{Proof outline.} 
We adapt the proof method of \cite[Theorem 1.7]{FLS2} to a new class of cutoff functions with compactly supported Fourier transforms that we introduce in Section \ref{ch prep lemm}. By the uncertainty principle, these functions have isolated zeros and are therefore not true cutoff functions and handling this point is a key technical challenge. As in \cite{FLS2}, we then define the Adiabatic Space Time Localization Observables (ASTLO) through second quantization with respect to suitably spacetime localized members of this function class. The key challenge is to derive a recursive monotonicity estimate (RME) which controls the time evolution of these observables relative to their initial expectation value. As in \cite{FLS2}, an important step in the proof is to separate the $x$ and $y$ dependence using the Cauchy-Schwartz inequality, this is possible thanks to the symmetry of the expansion. However, here this requires a Taylor expansion which is challenging in two ways compared to previous similar works:
\begin{itemize}
\item Since we need to absorb an $N$-dependent growing term, we need to perform that Taylor expansion to an extremely large order that diverges as the system size goes to infinity. For this, we use tight control on the the supremum norm of derivatives of our functions with bounded Fourier transforms.
\item The more challenging part of the RME is that we need a function class that is non-negative with non-negative derivatives and that, up to upper bounds with universal constants is preserved under differentiation. Ensuring these crucial aspects requires us to use a rather special choice of cutoff functions, essentially integrals over $\check\varphi^4$ with $\varphi\in C_c^\infty$. 
\end{itemize}
Our choice of cutoff function enables us to derive in \lemref{lemma symmetrized taylor} a symmetrized Taylor expansion which expresses these functions in terms of the admissible ones, which plays a fundamental role in managing the time evolution of the ASTLOs.  We believe that this class may be useful for constructing ASTLOs in other problems where very high derivatives need to be controlled. We are able to iterate the RME to enhance the dependency on the on the order of the Taylor expansion, $n$, leading to \corref{cor: iter bound}. 

The price we pay for the strong control on the derivatives of this class of cutoff functions compared to \cite{FLS2} is that we lose the compact support due to the uncertainty principle. Hence, we have to adapt the whole proof to this new setting, e.g., when we connect the averages of the ASTLO to those of the number operator, which is done in \thmref{N and ASLTO}. The main ingredients for this are the particular form of the functions we consider, especially the fast decay and the bound on the supremum norm, and the assumption of controlled density at the initial times. Applying this result to the one in \corref{cor: iter bound} we are able to obtain a bound that only depends on averages of the ASTLOs and the number operator at the initial time. 

At this point we can control the decay of the bound for $n$ that goes to infinity and in  \corref{bound on int with only initial time} we find some conditions on the adiabatic parameter that ensure that in the limit all the terms that depend on the number operator vanish. The whole construction  and the assumption of short range interactions are crucial to achieve this kind of decay.  Applying this result together with \thmref{N and ASLTO} to \thmref{theorem bound on integral of exp of N} gives us the desired bound for the special case of spheres. To control the number of particles in any bounded set $X$ we embed it in a ball and apply the previous result. \\ 

\textbf{Organization.} 
In Section \ref{ch prep lemm}, we introduce the new tailor-made class of band-limited cutoff functions, followed by two lemmas that give a bound on their derivatives and one about the symmetrized Taylor expansion. 
In Section \ref{ch bound astlo} we obtain a bound on the evolved ASTLO followed by a corollary obtained by iterating it. In Section \ref{main res} we establish a connection between the averages of the ASTLOs and those of the number operator restricted to a finite region. We then apply this result and obtain the bound that only contain averages at initial time. In this same section we give the conditions on the adiabatic parameter to ensure the decay. We conclude the proof in Section \ref{true main res} where we derive the bound for the special case of balls and we generalize it to bounded sets.

% \cml{To do:
% \begin{enumerate}
% \item Change proof from $\omega(n_x)=\delta$ to $\delta^{-1}\leq  \om(n_x)\leq \delta$. (Just replace $\delta$ by $\delta^{-1}$ in lower bounds.)
% \item Carla, can you confirm the definition of $|x|_{ts}$ to involve $d_X$ instead of $|x|$ (similar to \cite{SiZh}) we get the result as stated above. This would be geometrically nice as balls centered at the origin are not special. \textit{If} the proof adapts rather directly to that case, then you don't to replace $|x|$ by $d_X$ everywhere and we can just add a comment on this and keep the $|x|$ throughout.
% \item Please double-check if there are any other changes of convention/slight changes in result formulation that we need to address in the argument below (possibly just by a short remark if they are very obvious), to really get the theorem as stated above.
% \end{enumerate}
% }
% 

\subsection{Preliminaries}\label{ch prep lemm}

%Here we will construct a new class of functions with very important properties that will be essential for deriving the desired bound. They will also allow us to define the so-called ASTLO, operators which help us to track particles through the lattice. We will then derive some properties of the functions considered and use these results to prove an important lemma that will play a central role in controlling the evolution of ASTLOs.

\subsubsection{Function Class and ASTLO}
Consider now the following function classes
\begin{align*}
     &\cF=\left\{\frac{\int_{-\infty}^x{\check g}}{\int_{-\infty}^\infty \check g}\,:\, \check g=\check \varphi^4, \,\varphi\in C_c^\infty\left (B_1(0)\right)\right\},\\
     &\cG=\left\{ u=\check \psi^2\text{ for some } \psi\in C_c^{\infty}\left (B_1(0)\right)\right\}.
\end{align*}
   \\
   Given a function $f\in\cF$, then there exists $\varphi\in C^{\infty}_c(B_1(0))$ such that $f'=\check\varphi^4>0$. Then we can define the following functions

       \begin{align}
       &u:=\sqrt{f'}=\check\varphi^2,\label{u}\\
       &u_k:=\left(\check\varphi^{(k)}\right)^2 ,\label{uk}\\
       &f_k:=\frac{\int_{-\infty}^xu_k^2}{\int_{-\infty}^\infty u_k^2}\quad \text{for all }k\in \mathbb{N}.\label{fk}
   \end{align}
   We notice that a function $f_k$ is still in the function class $\cF$. In fact, since $\varphi\in C_c^{\infty}\left (B_1(0)\right)$ then $\widehat{\check{\varphi}^{(k)}}(\xi)\propto \xi^k \varphi(\xi)\in C_c^{\infty}\left (B_1(0)\right)$, which implies that $u_k\in \cG$.\\
For such a function, $f\in\cF$, we can also define the following operator, which in \cite{FLS2,FLS} was referred to as ASTLO (adiabatic space time localization observable),
\begin{equation*}
N_{f,ts}:= d\Gamma\big(f(|x|_{ts})\big)\:=\:\sum_{x\in\Lambda}f(|x|_{ts})\:a_x^{*}a_x,    
\end{equation*}
where
\begin{equation*}
    |x|_{ts}:=\frac{\rho-vt-|x|}{s}.
\end{equation*}
%\ccr{I don't think I need to define the astlo for a general set X because in the proof I only use the astlo defined as in the previous line. Only in the last step I generalized it to bounded sets X, but I don't need to redifine the astlo}

\subsubsection{Properties of the Function Class}

In this section we show two important properties of elements of  the two function classes defined above. A key ingredient in the the proof is their particular structure, for instance we can bound their supremum norm thanks to the fact that their Fourier transform is compactly supported.

\begin{lemma}[Upper bounds on the supremum norm of the derivatives]\label{Upper bounds on the derivatives}
 For any $f\in\cF$ and for any $ 1\leq k,j\in\mathbbm{N}$, there exists a constant $C_f$, such that the following hold:
\begin{align}
    \label{bound der u} 
    &\quad||u_j^{(k)}||_{\infty}\leq \,2^k C_f,\\
    \label{bound der f}
    &\quad||f_j^{(k)}||_{\infty}\leq  \,4^{k-1}C_f. 
\end{align}

\end{lemma}
\begin{proof}
1. We start showing inequality \eqref{bound der u}. Because of how $u_j$ is defined we obtain the following
    \begin{equation}\label{ineq 1}
        ||u_j^{(k)}||_{\infty}\leq C ||\widehat{u_j^{(k)}}||_1=C||\xi^k\widehat{\left(\check\varphi^{(j)}\right)}\ast\widehat{\left(\check\varphi^{(j)}\right)}||_1=C||\xi^k\left(\xi^j\varphi(\xi)\right)\ast \left(\xi^j\varphi(\xi)\right)||_1.
    \end{equation}
  Since $\varphi$ is supported on $B_1(0)$ 
  \begin{equation}\label{ineq 2}
      |\left(\xi^j\varphi\right)\ast \left(\xi^j\varphi\right)|\,=\big|\int_{\mathbbm{R}}\xi^j\varphi(\xi)(\xi-\nu)^j\varphi(\xi-\nu)d\xi\big|\leq \int_{\mathbbm{R}}\big|\varphi(\xi)\big|\big|\varphi(\xi-\nu)\big|d\xi = |\varphi|\ast |\varphi|,
  \end{equation}
  and since $|\varphi|\ast |\varphi|$ is supported on $B_2(0)$, applying \eqref{ineq 2} to \eqref{ineq 1}, it implies that 
  \begin{equation*}
    ||u_j^{(k)}||_{\infty}\leq  ||\xi^k\left(\xi^j\varphi(\xi)\right)\ast \left(\xi^j\varphi(\xi)\right)||_1\leq 2^k|||\varphi|\ast |\varphi|||_1.
  \end{equation*}
 2. To prove inequality \eqref{bound der f} we observe that
    \begin{align*}
       ||f_j^{(k)}||_{\infty}&\leq C||\widehat{f_j^{(k)}}||_1=C||\xi^{k-1} \widehat{f_j'}(\xi)||_1=C||\xi^{k-1} \widehat{(\check\varphi^{(j)})^4}(\xi)||_1\\
       &=C||\xi^{k-1} (\xi^j\varphi(\xi))\ast(\xi^j\varphi(\xi))\ast(\xi^j\varphi(\xi))\ast(\xi^j\varphi(\xi))||_1
    \end{align*}
    As before, since $\varphi$ is supported on $B_1(0)$ we obtain
    \begin{align*}
        |(\xi^j\varphi(\xi))\ast(\xi^j\varphi(\xi))\ast(\xi^j\varphi(\xi))\ast(\xi^j\varphi(\xi))|\leq |\varphi|\ast |\varphi|\ast |\varphi|\ast |\varphi|
    \end{align*}
   and that $\xi^{k-1}\varphi\ast\varphi\ast\varphi\ast\varphi(\xi)$ is supported on $B_4(0)$. Together they imply that 

    \begin{equation*}
     ||f_j^{(k)}||_{\infty}\leq ||\xi^{k-1} |\varphi|\ast |\varphi|\ast |\varphi|\ast |\varphi|(\xi)||_1\leq 4^{k-1} ||\varphi\ast\varphi\ast\varphi\ast\varphi(\xi)||_1,
    \end{equation*}
    which leads to the desired bound.\\
   
\end{proof}

\begin{lemma}[Bounding derivatives with products of admissible functions]\label{lemma bounding derivatives in term of phi}
    For any $f\in\cF$ and any integer $k\ge 1$, the following inequalities hold:
    \begin{align}
       &| f^{(k)}(x)| \leq 4^{k} \cdot k^2\sum_{a=0}^{k}\Big(u_a(x)u_a(y)\,+\,u_a(x)(u_a(x)-u_a(y))\Big), \label{bound on f} \\
       & | u(x)u^{(k)}(x)|\leq\: 2^{k-1}(k+1)\sum_{a=0}^{k}\Big(u_a(x)u_a(y)+u_a(x)\big(u_a(x)-u_a(y)\big)\Big),\label{bound on u}
    \end{align}

    where $u_a$ is defined as in \eqref{uk}.
\end{lemma}
\begin{proof}
1.   Because $f\in\cF$, then 
\begin{equation*}
    f=\frac{\int_{-\infty}^x{\check\varphi^4}}{\int_{-\infty}^\infty \check\varphi^4}= C_{\varphi}\int_{-\infty}^x{\check\varphi^4}\quad  \text{ for some } \varphi\in C_c^\infty\left (B_1(0)\right),
\end{equation*}
then its first derivative is given by $f'(x)=\check\varphi^4(x)\geq 0$. This implies that the $k$-th derivatives of f is just going to be the sum on $4^{k}$ terms of the kind $\check\varphi^{(a)}\check\varphi^{(b)}\check\varphi^{(c)}\check\varphi^{(d)}$, such that $a+b+c+d=k$.\\
Since for any real number $\lambda_1,\lambda_2,\lambda_3,\lambda_4$ it holds that
\begin{equation*}
   \lambda_1\lambda_2\lambda_3\lambda_4\leq \frac{1}{4}\Big(\lambda_1^4+\lambda_2^4+\lambda_3^4+\lambda_4^4\Big),
\end{equation*}
we obtain
\begin{align*}
   | \check\varphi^{(a)}(x)&\check\varphi^{(b)}(x)\check\varphi^{(c)}(x)\check\varphi^{(d)}(x)|\\
   \leq& \frac{1}{4}\Big((\check\varphi^{(a)})^4(x)+(\check\varphi^{(b)})^4(x)+(\check\varphi^{(c)})^4(x)+(\check\varphi^{(d)})^4(x)\Big)\\
    =&\frac{1}{4}\Big(u_a(x)u_a(x)+u_b(x)u_b(x)+u_c(x)u_c(x)+u_d(x)u_d(x)\Big)\\
    =&\frac{1}{4}\Big(u_a(x)u_a(y)+u_b(x)u_b(y)+u_c(x)u_c(y)+u_d(x)u_d(y)\Big)\\
    &+\frac{1}{4}\Big[u_a(x)(u_a(x)-u_a(y))+u_b(x)(u_b(x)-u_b(y))\nonumber\\
    &+u_c(x)(u_c(x)-u_c(y))+u_d(x)(u_d(x)-u_d(y))\Big].\nonumber
\end{align*}
This implies
\begin{align*}
   | f^{(k)}(x)|&\leq 4^{k} \sum_{\substack{a,b,c,d:\\ a+b+c+d=k}}|\check\varphi^{(a)}\check\varphi^{(b)}\check\varphi^{(c)}\check\varphi^{(d)}|\\ 
    &\leq 4^{k}\sum_{\substack{a,b,c,d:\\ a+b+c+d=k}}\Big(u_a(x)u_a(y)+u_b(x)u_b(y)+u_c(x)u_c(y)+u_d(x)u_d(y)\Big)\\
    & +\,4^{k}\sum_{\substack{a,b,c,d:\\ a+b+c+d=k}}\Big(u_a(x)(u_a(x)-u_a(y))+u_b(x)(u_b(x)-u_b(y))\nonumber\\
    &\,+u_c(x)(u_c(x)-u_c(y))+u_d(x)(u_d(x)-u_d(y))\Big)\nonumber\\
    &=\, 4\cdot4^{k-1}\sum_{\substack{a,b,c,d:\\ a+b+c+d=k}}\Big(u_a(x)u_a(y)\,+\,u_a(x)(u_a(x)-u_a(y))\Big)\\
    &=4^{k} \sum_{a=0}^{k}\Big(u_a(x)u_a(y)\,+\,u_a(x)(u_a(x)-u_a(y))\Big)\sum_{\substack{b,c,d:\\ b+c+d=k-a}}1\\
    &\leq 4^{k} \sum_{a=0}^{k}(k-a)^2\Big(u_a(x)u_a(y)\,+\,u_a(x)(u_a(x)-u_a(y))\Big) \\
    & \leq 4^{k} \cdot k^2\sum_{a=0}^{k}\Big(u_a(x)u_a(y)\,+\,u_a(x)(u_a(x)-u_a(y))\Big)
\end{align*}
 
     2. Similarly, since $u=\sqrt{f'}=\check\varphi^2$, its $k$-th derivative is the sum of $2^{k}$ elements of the form $\check\varphi^{(a)}\check\varphi^{(b)}$, with $a+b=k$. This way the product $u(x)u^{(k)}$ is given by $2^{k}$ summands of the form
  \begin{align*}
     |u(x)\check\varphi^{(a)}(x)\check\varphi^{(b)}(x)|& \leq \frac{1}{4}\Big(2u^2(x)+(\check\varphi^{(a)})^4(x)+(\check\varphi^{(b)})^4(x)\Big)\\
     &=\frac{1}{4}\Big(2u(x)u(y)+u_a(x)u_a(y)+u_b(x)u_b(y)\Big)\\
     &+\frac{1}{4}\Big(2u(x)\big(u(x)-u(y)\big)+u_a(x)\big(u_a(x)-u_a(y)\big)+u_b(x)\big(u_b(x)-u_b(y)\big)\Big),\nonumber
 \end{align*}
again for some $u_a, u_b\in\cG$. So we can bound $u\cdot u^{(k)} $ by
\begin{align*}
   | u(x)&u^{(k)}(x)|\\
    \leq& \:2^{k}\sum_{\substack{a,b:\\ a+b=k}}u(x)|\check\varphi^{(a)}(x)\check\varphi^{(b)}(x)|\\
    \leq&\: 2^{-2}\cdot 2^{k}\sum_{\substack{a,b:\\ a+b=k}}\Big(2u(x)u(y)+u_a(x)u_a(y)+u_b(x)u_b(y)\\
     &+2u(x)\big(u(x)-u(y)\big)+u_a(x)\big(u_a(x)-u_a(y)\big)+u_b(x)\big(u_b(x)-u_b(y)\big)\Big)\nonumber\\
     \leq& \:2^{(k-1)}\Big( u(x)u(y)+u(x)\big(u(x)-u(y)\big)\Big)\sum_{\substack{a,b:\\ a+b=k}}1\\
     &+2^{(k-2)}\cdot 2 \sum_{a=0}^{k}\Big(u_a(x)u_a(y)+u_a(x)\big(u_a(x)-u_a(y)\big)\Big)\sum_{b:\,b=k-a}1\nonumber\\
     =&\:2^{(k-1)}(k)\Big( u(x)u(y)+u(x)\big(u(x)-u(y)\big)\Big)\\
     &+2^{(k-1)}\sum_{a=0}^{k}\Big(u_a(x)u_a(y)+u_a(x)\big(u_a(x)-u_a(y)\big)\Big)\nonumber\\
     \leq&\: 2^{(k-1)}(k+1)\sum_{a=0}^{k}\Big(u_a(x)u_a(y)+u_a(x)\big(u_a(x)-u_a(y)\big)\Big)
\end{align*}
which finishes the proof.
\end{proof}

\lemref{lemma bounding derivatives in term of phi} allows us to bound the derivatives of the admissible functions by sums of products of functions in the same class. It is very important to notice that these sums are made by a part which is symmetric in the x and y components and one that is not. This fact is crucial and will play a central role in the proof of the next lemma. 
\bigskip

\subsubsection{Symmetrized Taylor expansion}

We state a lemma that allows us to bound the difference $f(x)-f(y)$ by a sum of symmetrized terms of the form $u_a(x)u_a(y)$. The main tools for the proof are  \lemref{Upper bounds on the derivatives} and  \lemref{lemma bounding derivatives in term of phi}. We make strong use of the properties of the function in our class. The idea is to bound their derivatives by symmetrized terms, which we keep as they are, and non-symmetrized ones, which we will further Taylor expand. By iterating this procedure until all the terms are symmetrized we can prove the desired expansion. This lemma allows us to separate the x and y dependency and will be fundamental for controlling the time evolution of the ASTLOs.
\\
Let us first define the following index sets for any integer $l\ge 1$:

\begin{align*}
    &\cK_{l,m}:=\{\underline k=(k_1,\dots,k_l)\in\N^l \; \text{ such that } k_i\ge 1 \text{ for all }i\in[1,l] \text{ and } k_1+\dots+k_l=m\}\\
    &\cJ_{\und k}:=\{\und j:=(j_1,\dots,j_l)\in\N^l\text{ such that } 0\le j_i\le k_i \text{ for all } i\in[1,l]\}\quad \text{for every }\und k\in \cK_{l,m}.
\end{align*}
For a vector $\und j\in \cJ_{\und k}$ we can also define the following quantity

\begin{align}
    J_l:=\sum_{i=1}^l j_i.\label{Jl}
\end{align}

Recall that $u_k:=\left(\check\varphi^{(k)}\right)^2$. We come to the following lemma.
\begin{lemma}[Symmetrized Taylor Expansion]\label{lemma symmetrized taylor}
 Let $f\in\cF$. For all $n\in\mathbb{N}$, there exist two positive constants $C$ and $C_f$ such that
 \begin{align*}
      |f&(x)-f(y)|\\
     &\leq|x-y|u(x)u(y)\\
     &\quad+C|x-y|\sum_{l=1}^{n-1}\sum_{m=l}^n 4^{m-l}|x-y|^{m} \sum_{\und k\in \cK_{l,m}} \;\sum_{\und j\in \cJ_{\und k}}\frac{u_{J_l}(x)u_{J_l}(y)}{(k_1-1)!\dots (k_l-1)! }\\
     &\quad+C_f|x-y|^{n+1}4^n\Big(\frac{1}{n!}+\sum_{l=2}^{n-1}\left(\frac{e}{2}\right)^l\sum_{m=l-1}^n  2^m\Big),
 \end{align*}
 where the sums should be dropped for $n=1$.

\end{lemma}

\begin{proof}
    We start by writing the Taylor expansion
    \begin{equation}
    f(x)-f(y)=\sum_{k_1=1}^n \frac{(x-y)^{k_1}}{k_1!}f^{(k_1)}(x)\;+\;\frac{(x-y)^{n+1}}{(n+1)!}f^{(n+1)}(a_{0_1}),
    \end{equation}
where we wrote the remainder in the Lagrange form and $a_{0_1}$ is some number between $x$ and $y$.\\
The first term of the expansion can be rewritten as
\begin{equation}
    (x-y)f'(x)=(x-y)u(x)u(x)=(x-y)u(x)u(y)+(x-y)u(x)\big(u(x)-u(y)\big).
\end{equation}
If $n=1$, then the proof is concluded. If $n>1$ we can Taylor expand the factor $u(x)-u(y)$
\begin{align*}
   (x-y)&u(x)\big(u(x)-u(y)\big)\\
   &=(x-y)u(x)\sum_{j=1}^{n-1} \frac{(x-y)^j}{j!}u^{(j)}\;+\;u(x)\frac{(x-y)^{n+1}}{n!}u^{(n)}(a_{1_1})\\
   &=\sum_{j=1}^{n-1} u(x)\frac{(x-y)^{(j+1)}}{j!}u^{(j)}\;+\;u(x)\frac{(x-y)^{n+1}}{n!}u^{(n)}(a_{1_1})\\
   &=\sum_{k_1=2}^n u(x)\frac{(x-y)^{(k_1)}}{(k_1-1)!}u^{(k_1-1)}\;+\;u(x)\frac{(x-y)^{n+1}}{n!}u^{(n)}(a_{1_1}),
\end{align*}
again for some number $a_{1_1}$ between $x$ and $y$.\\
So we obtain
\begin{align*}
        f(x)-f(y)&=(x-y)u(x)u(y)+\sum_{k_1=2}^n (x-y)^k\Big(\frac{f^{(k_1)}(x)}{k_1!}\;+\;u(x)\frac{u^{(k_1-1)}(x)}{(k_1-1)!}\Big)\;\\
        &+\;(x-y)^{n+1}\Big[\frac{f^{(n+1)}(a_{0_1})}{(n+1)!}+\frac{u(x)u^{(n)}(a_{1_1})}{n!}\Big]\nonumber\\
        =&(x-y)u(x)u(y)+\sum_{k_1=2}^n (x-y)^k\Big(\frac{f^{(k_1)}(x)}{k_1!}\;+\;u(x)\frac{u^{(k_1-1)}(x)}{(k_1-1)!}\Big)\;\\
        &+\;\frac{(x-y)^{n+1}}{(n+1)!}\Big[f^{(n+1)}(a_{0_1})+(n+1)u(x)u^{(n)}(a_{1_1})\Big],\nonumber
    \end{align*}
which implies 

\begin{align}\label{f(x)-f(y)}
     |f(x)-f(y)|
     \leq&|x-y|u(x)u(y)\\
     &+\sum_{k_1=2}^n |x-y|^{k_1}\Big(\frac{|f^{(k_1)}(x)|}{k_1!}\;+\;u(x)\frac{|u^{(k_1-1)}(x)|}{(k_1-1)!}\Big)\nonumber\\
    &+\;\frac{|x-y|^{n+1}}{(n+1)!}\Big[|f^{(n+1)}(a_{0_1})|+(n+1)|u(x)u^{(n)}(a_{1_1})|\Big].\nonumber
\end{align}

Applying \eqref{bound on f} and \eqref{bound on u} we can estimate the sum in inequality \eqref{f(x)-f(y)} in the following way

\begin{align*}
 \sum_{k_1=2}^n &|x-y|^{k_1}\Big(\frac{|f^{(k_1)}(x)|}{k_1!}\;+\;u(x)\frac{|u^{(k_1-1)}(x)|}{(k_1-1)!}\Big)\\
 &\leq\sum_{k_1=2}^n |x-y|^{k_1}\Big[\frac{4^{k_1} \cdot k_1^2}{k_1!}\sum_{{j_1}=0}^{k_1}(u_{j_1}(x)u_{j_1}(y)\,+\,u_{j_1}(x)(u_{j_1}(x)-u_{j_1}(y)))\\
&\qquad+\frac{2^{k_1-2}k_1}{(k_1-1)!}\sum_{{j_1}=0}^{k_1-1}\Big(u_{j_1}(x)u_{j_1}(y)+u_{j_1}(x)\big(u_{j_1}(x)-u_{j_1}(y)\big)\Big)\Big]\nonumber\\
&\leq\sum_{k_1=2}^n \frac{2\cdot4^{k_1}k_1}{(k_1-1)!}|x-y|^{k_1}\sum_{{j_1}=0}^{k_1}\Big(u_{j_1}(x)u_{j_1}(y)\,+\,u_{j_1}(x)\big(u_{j_1}(x)-u_{j_1}(y)\big)\Big)\\
&\leq \sum_{k_1=2}^n \frac{2\cdot4^{k_1}k_1}{(k_1-1)!}|x-y|^{k_1}\sum_{{j_1}=0}^{k_1}u_{j_1}(x)u_{j_1}(y)\\
&\qquad+\sum_{k_1=2}^n \frac{4^{k_1}k_1}{(k_1-1)!}|x-y|^{k_1}\sum_{{j_1}=0}^{k_1}u_{j_1}(x)\big(u_{j_1}(x)-u_{j_1}(y)\big).\nonumber
\end{align*}
We obtained a first part of symmetric terms and a second one containing terms of the form $u_{j_1}(x)\big(u_{j_1}(x)-u_{j_1}(y)\big)$ that we can further Taylor expand and treat in the same way as the term $u(x)\big(u(x)-u(y)\big)$
\begin{align*}
    \big|u_{j_1}(x)\big|&\big|u_{j_1}(x)-u_{j_1}(y)\big|\\ 
    &\leq\sum_{k_2=1}^{n-k_1} \frac{|x-y|^{k_2}}{k_2!}|u_{j_1}(x)u^{(k_2)}(x)|\;+\;\;\frac{|x-y|^{n-k_2+1}}{(n-k_2+1)!}|u_{j_1}(x)u_{j_1}^{n-k_1+1}(a_{2_1})|\\
    &\leq\sum_{k_2=1}^{n-k_1} \frac{|x-y|^{k_2}}{(k_2-1)!}2^{(k_2-2)}\sum_{j_2=0}^{k_2}\Big(u_{j_1+j_2}(x)u_{j_1+j_2}(y)+u_{j_1+j_2}(x)\big(u_{j_1+j_2}((x)-u_{j_1+j_2}((y)\big)\Big)\\
        &\qquad+\;\frac{|x-y|^{n-k_1+1}}{(n-k_1+1)!}|u_{j_1}(x)u_{j_1}^{n-k_1+1}(a_{2_{j_1}})|.\nonumber
    \end{align*}

After iterating the process on the non symmetrized terms  we obtain the following inequality

\begin{align}\label{eq after it taylor}
     &|f(x)-f(y)|\\
     &\leq|x-y|u(x)u(y)+\sum_{k_1=2}^n\sum_{j_1=0}^{k_1}\frac{4^{k_1}k_1}{(k_1-1)!}|x-y|^{k_1}u_{j_1}(x)u_{j_1}(y)\nonumber\\
     &+\sum_{l=2}^{n-1}\sum_{k_1=2}^n\sum_{j_1=0}^{k_1}\sum_{k_2=1}^{n-k_1}\sum_{j_2=0}^{k_2}...\sum_{k_l=1}^{n-k_1-...-k_{l-1}}\sum_{j_l=0}^{k_l} \frac{4^{k_1}k_1\cdot 2^{(k_2-2)+\dots +(k_l-2)} }{(k_1-1)!\dots (k_l-1)! }|x-y|^{k_1+\dots +k_l}u_{J_l}(x)u_{J_l}(y)\nonumber\\
     &+\frac{|x-y|^{n+1}}{(n+1)!}\Big[|f^{(n+1)}(a_{0_1})|+(n+1)|u(x)u^{(n)}(a_{1_1})|\Big]\nonumber\\
     &+\sum_{l=2}^{n-1}\sum_{k_1=2}^n\sum_{j_1=0}^{k_1}\sum_{k_2=1}^{n-k_1}\sum_{j_2=0}^{k_2}...\sum_{k_{l-1}=1}^{n-k_1-...-k_{l-2}} \frac{4^{k_1}k_1 2^{(k_2-2)+\dots +(k_{l-1}-2)} }{(k_1-1)!\dots (k_{l-1}-1)! }|x-y|^{n+1}u_{J_{l-1}}(x)u_{J_{l-1}}^{(n-k_1-\dots k_{l-1}+1)}(a_j),\nonumber
\end{align}
where $a_j$ are point between $x$ and $y$ and $J_l$ is defined as in \eqref{Jl}.
By rescaling $k_1\to k_1-1$, we obtain:

\begin{align}
\sum_{k_1=2}^n&\sum_{j_1=0}^{k_1}\frac{4^{k_1}k_1}{(k_1-1)!}|x-y|^{k_1}u_{j_1}(x)u_{j_1}(y)\nonumber\\
  & =|x-y|\sum_{k_1=1}^{n-1}\sum_{j_1=0}^{k_1}\frac{4\cdot4^{k_1}(k_1+1)}{k_1!}|x-y|^{k_1}u_{j_1}(x)u_{j_1}(y)\nonumber\\
  &\leq |x-y|\sum_{k_1=1}^{n}\sum_{j_1=0}^{k_1}\frac{4\cdot4^{k_1}2k_1}{k_1!}|x-y|^{k_1}u_{j_1}(x)u_{j_1}(y)\nonumber\\
  &= |x-y|\sum_{k_1=1}^{n}\sum_{j_1=0}^{k_1}\frac{8\cdot4^{k_1}}{(k_1-1)!}|x-y|^{k_1}u_{j_1}(x)u_{j_1}(y).\label{control on first long sum}
\end{align}
Similarly, after introducing the parameter $m:=k_1+\dots+k_l\in[l,n]$, the second sum becomes
\begin{align*}
   \sum_{l=2}^{n-1}&\sum_{k_1=2}^n\sum_{j_1=0}^{k_1}\sum_{k_2=1}^{n-k_1}\sum_{j_2=0}^{k_2}...\sum_{k_l=1}^{n-k_1-...-k_{l-1}}\sum_{j_l=0}^{k_l} \frac{4^{k_1}k_1\cdot 2^{(k_2-2)+\dots +(k_l-2)} }{(k_1-1)!\dots (k_l-1)! }|x-y|^{k_1+\dots +k_l}u_{J_l}(x)u_{J_l}(y)\\ 
   \leq &|x-y|\sum_{l=2}^{n-1}\sum_{m=l}^n\sum_{k_1=1}^n\sum_{j_1=0}^{k_1}\sum_{k_2=1}^{n-k_1}\sum_{j_2=0}^{k_2}...\sum_{k_l=1}^{n-k_1-...-k_{l-1}}\sum_{j_l=0}^{k_l} \frac{8\cdot4^{k_1}\cdot 2^{m-k_1-2(l-1)} }{(k_1-1)!\dots (k_l-1)! }|x-y|^{m}u_{J_l}(x)u_{J_l}(y)\\
   \leq&|x-y|\sum_{l=2}^{n-1}\sum_{m=l}^n\sum_{k_1=1}^n\sum_{j_1=0}^{k_1}\sum_{k_2=1}^{n-k_1}\sum_{j_2=0}^{k_2}...\sum_{k_l=1}^{n-k_1-...-k_{l-1}}\sum_{j_l=0}^{k_l} \frac{32\cdot4^{k_1}\cdot 4^{m-k_1-l} }{(k_1-1)!\dots (k_l-1)! }|x-y|^{m}u_{J_l}(x)u_{J_l}(y)\\
   \leq & \,C|x-y|\sum_{l=2}^{n-1}\sum_{m=l}^n 4^{m-l}|x-y|^{m} \sum_{\und k\in \cK_{l,m}} \;\sum_{\und j\in \cJ_{\und k}}\frac{u_{J_l}(x)u_{J_l}(y)}{(k_1-1)!\dots (k_l-1)! }.
\end{align*}
This, together with inequality \eqref{control on first long sum} gives us
\begin{align*}
    |x-y|\sum_{k_1=1}^{n}\sum_{j_1=0}^{k_1}&\frac{8\cdot4^{k_1}}{(k_1-1)!}|x-y|^{k_1}u_{j_1}(x)u_{j_1}(y)\\
    +&C|x-y|\sum_{l=2}^{n-1}\sum_{m=l}^n 4^{m-l}|x-y|^{m} \sum_{\und k\in \cK_{l,m}} \;\sum_{\und j\in \cJ_{\und k}}\frac{u_{J_l}(x)u_{J_l}(y)}{(k_1-1)!\dots (k_l-1)! }\nonumber\\
    \leq C|x-y|&\sum_{l=1}^{n-1}\sum_{m=l}^n 4^{m-l}|x-y|^{m} \sum_{\und k\in \cK_{l,m}} \;\sum_{\und j\in \cJ_{\und k}}\frac{u_{J_l}(x)u_{J_l}(y)}{(k_1-1)!\dots (k_l-1)! }.
\end{align*}

Analogously we can bound the fourth summand of inequality \eqref{eq after it taylor}

\begin{align}
    \sum_{l=2}^{n-1}&\sum_{k_1=2}^n\sum_{j_1=0}^{k_1}\sum_{k_2=1}^{n-k_1}\sum_{j_2=0}^{k_2}...\sum_{k_l=1}^{n-k_1-...-k_{l-1}} \frac{4^{k_1}k_1\cdot 2^{(k_2-2)+\dots +(k_{l-1}-2)} }{(k_1-1)!\dots (k_{l-1}-1)! }|x-y|^{n+1}u_{J_{l-1}}(x)u_{J_{l-1}}^{(n-k_1-\dots k_{l-1}+1)}(a_j)\nonumber\\
    &\leq \,C|x-y|^{n+1} \sum_{l=2}^{n-1}\sum_{m=l-1}^n4^{m-l}\sum_{\und k\in \cK_{l,m}} \;\sum_{\und j\in \cJ_{\und k}}\frac{u_{J_{l-1}}(x)u_{J_{l-1}}^{(n-m+1)}(a_j) }{(k_1-1)!\dots (k_{l-1}-1)! },\label{ineq sums of lagrange rem}
\end{align}
for some positive constant $C$. Now, using inequality \eqref{bound der u}

\begin{align*}
   \sum_{\und k\in \cK_{l,m}}& \;\sum_{\und j\in \cJ_{\und k}}\frac{u_{J_{l-1}}(x)u_{J_{l-1}}^{(n-m+1)}(a_j) }{(k_1-1)!\dots (k_{l-1}-1)! }\\
    \leq&\sum_{\und k\in \cK_{l,m}} \;\sum_{\und j\in \cJ_{\und k}}\frac{C_f2^{n-m+1}}{(k_1-1)!\dots (k_{l-1}-1)! }\\
    =&C_f2^{n-m+1}\sum_{\und k\in \cK_{l,m}} \frac{k_1\cdot\dots\cdot k_{l-1}}{(k_1-1)!\dots (k_{l-1}-1)! }\\
    \leq & C_f2^{n-m+1}\left(2e\right)^{l-1}.
\end{align*}
This allows us to bound \eqref{ineq sums of lagrange rem} in the following way
\begin{align*}
    C|x-y|^{n+1}& \sum_{l=2}^{n-1}\sum_{m=l-1}^n4^{m-l}\sum_{\und k\in \cK_{l,m}} \;\sum_{\und j\in \cJ_{\und k}}\frac{u_{J_{l-1}}(x)u_{J_{l-1}}^{(n-m+1)}(a_j) }{(k_1-1)!\dots (k_{l-1}-1)! }\\
    \leq&C_f|x-y|^{n+1} \sum_{l=2}^{n-1}\sum_{m=l-1}^n4^{m-l}2^{n-m+1}\left(2e\right)^{l-1}\\
    =& C_f 2^{n}|x-y|^{n+1} \sum_{l=2}^{n-1}\left(\frac{e}{2}\right)^l\sum_{m=l-1}^n  2^m.
\end{align*}

Applying inequalities \eqref{bound der u} and \eqref{bound der f}
 on the third summand of inequality \eqref{eq after it taylor} we obtain 
 
 \begin{align*}
     \frac{|x-y|^{n+1}}{(n+1)!}&\Big[|f^{(n+1)}(a_{0_1})|+(n+1)|u(x)u^{(n)}(a_{1_1})|\Big]\\
     \leq & \frac{C_f|x-y|^{n+1}}{(n+1)!}\left(4^n+(n+1)2^n\right)\\
      \leq & C_f\frac{4^n|x-y|^{n+1}}{n!},
 \end{align*}
and this concludes the proof.
 \end{proof}

\subsection{Control over the time evolved ASTLO}\label{ch bound astlo}

Employing  \lemref{lemma symmetrized taylor}, the properties of the ASTLOs, and the structure of the Hamiltonian we can derive an important bound of the time evolved ASTLO. We exploit the fact that we can separate the x and y dependence using the Cauchy-Schwartz inequality to obtain a bound that depends on the average of the ASTLO itself.

 \begin{theorem}[Bound on the evolved ASTLO]\label{theorem bound on integral of exp of N}
     For all $f\in\cF$, $ n\geq1$, and $c>\kappa$ there exist two positive constants $C,\,C_{f}$ such that, for all $s, t, \rho>0$, the following holds
  \begin{align}\label{lemma before recursion}
       \langle N_{f,ts}\rangle_t\leq  \langle &N_{f,ts}\rangle_0+\frac{\kappa-v}{s}\int_0^t \langle\: N_{f',rs}\rangle_r dr\\
    &+C\sum_{l=1}^{n-1}\sum_{m=l}^n \sum_{\und k\in \cK_{l,m}} \;\sum_{\und j\in \cJ_{\und k}}\frac{4^{m-l}\cdot s^{-m-1}}{(k_1-1)!\dots (k_l-1)! }\kappa^{(m)}\int_0^t\big\langle N_{f'_{J_l,rs}}\big\rangle_rdr \nonumber   \\
    &+C_f4^n\Big(\frac{1}{n!}+\sum_{l=2}^{n-1}\left(\frac{e}{2}\right)^l\sum_{m=l-1}^n  2^m\Big)s^{-n-1}\kappa^{(n+1)}  t \big\langle N\big\rangle_0,\nonumber
  \end{align}
  where for $n=1$ the sums should be dropped.
 \end{theorem}
We notice that from definition \eqref{fk}, $f'_{J_l,rs}=u^2_{J_l}(|x|_{rs})$ and we define 
\begin{align*}
    \kappa^{(m)}:=\sup _{x\in\Lam} \sum_{y\in\Lam} \abs{h_{xy}}\abs{x-y}^m.
\end{align*}
 \begin{proof}

    Considering $N_{f,ts}$ as defined above for $t, s, \rho>0$, we observe that
    \begin{equation}\label{derivative exp value N}
      \frac{d}{dt}\langle\: N_{f,ts}\rangle_t\: =\: \langle i[H,N_{f,ts}] + \frac{\partial}{\partial t } N_{f,ts}\:\rangle_t .
    \end{equation}
    We will bound each of the terms on the right hand side of this equation separately. We first notice that the second term in equation \eqref{derivative exp value N} can be rewritten in the following way:
    \begin{equation}\label{deriv astlo}
    \frac{\partial}{\partial t}N_{f,ts}=\frac{\partial}{\partial t}\sum_{x\in\Lambda}f(|x|_{ts})n_x=-\frac{v}{s}\sum_{x\in\Lambda}f'(|x|_{ts})n_x=-\frac{v}{s}N_{f',ts}
\end{equation}
\\
Then, by \cite[lemma A.2]{FLS2} and  \lemref{lemma symmetrized taylor}, for any $n\in\mathbb{N}$, we can bound the first term on the right hand side of \eqref{derivative exp value N} 

\begin{align}\label{bound on comm}
    \Big|\Big\langle&\psi, i[H,N_{f,ts}]\psi\Big\rangle\Big|=\Big|\Big\langle\psi, \sum_{x,y\in\Lambda,\,x\neq y}J_{xy}\Big(f(|x|_{ts})-f(|y|_{ts})\Big)b_x^{*}b_y\psi\Big\rangle\Big|\\
    \leq& s^{-1}\sum_{\substack{x,y\in\Lambda,\\ \,x\neq y}}|J_{xy}||x-y|u(|x|_{ts})u(|y|_{ts})\big|\big\langle\psi, b_x^{*}b_y\psi\big\rangle\big|\\
     +&C\sum_{l=1}^{n-1}\sum_{m=l}^n \sum_{\und k\in \cK_{l,m}} \;\sum_{\und j\in \cJ_{\und k}}\frac{4^{m-l}s^{-m-1}}{(k_1-1)!.. (k_l-1)! }\sum_{\substack{x,y\in\Lambda,\\ \,x\neq y}}|J_{xy}||x-y|^{m+1}u_{J_l}(|x|_{ts})u_{J_l}(|y|_{ts})\big|\big\langle\psi, b_x^{*}b_y\psi\big\rangle\big|\nonumber\\
    +&C_f4^n\Big(\frac{1}{n!}+\sum_{l=2}^{n-1}\left(\frac{e}{2}\right)^l\sum_{m=l-1}^n  2^m\Big)s^{-n-1}\sum_{\substack{x,y\in\Lambda,\\ \,x\neq y}}|J_{xy}||x-y|^{n+1}\big|\big\langle\psi, b_x^{*}b_y\psi\big\rangle\big|,\nonumber
\end{align}
for any $\psi\in D(N^{\frac{1}{2}})$.\\
Using the Cauchy-Schwarz inequality we can bound the first term

\begin{align}
    \sum_{x,y\in\Lambda,\,x\neq y}&|J_{xy}||x-y|u(|x|_{ts})u(|y|_{ts})\big|\big\langle\psi, b_x^{*}b_y\psi\big\rangle\big|\nonumber\\
    =&  \sum_{x,y\in\Lambda,\,x\neq y}|J_{xy}||x-y|\big|\big\langle u(|x|_{ts})b_x\psi, u(|y|_{ts})b_y\psi\big\rangle\big|\nonumber\\
    \leq&  \:\Big(\sum_{x\in\Lambda}f'(|x|_{ts})\Big\langle \psi,b_x^{*}b_x\psi\Big\rangle\Big(\sum_{y\in\Lambda, \,y\neq x} |J_{xy}|||x-y|\Big)\Big)^{\frac{1}{2}}\nonumber\\
&\times\Big(\sum_{y\in\Lambda}f'(|x|_{ts})\Big\langle \psi,b_y^{*}b_y\psi\Big\rangle\Big(\sum_{x\in\Lambda, \,y\neq x} |J_{xy}|||x-y|\Big)\Big)^{\frac{1}{2}} \nonumber\\
\leq& \:\kappa\Big\langle\psi,\,N_{f'_{ts}}\psi\Big\rangle.\label{bound first order}
\end{align}
 We can apply the same reasoning to the higher order terms to obtain 
 
\begin{align}\label{bound higher order}
    \sum_{x,y\in\Lambda,\,x\neq y}&|J_{xy}||x-y|^{m}u_{J_l}(|x|_{ts})u_{J_l}(|y|_{ts})\big|\big\langle\psi, b_x^{*}b_y\psi\big\rangle\big|\leq\kappa^{(m)}\Big\langle\psi,\,N_{f'_{J_l,ts}}\psi\Big\rangle,
\end{align}
where we defined $f'_{J_l,ts}:=u^2_{J_l}(|x|_{ts})$.

Similarly we can bound the remainder term by

\begin{align}\label{bound remainder}
 \sum_{x,y\in\Lambda,\,x\neq y}|J_{xy}||x-y|^{n+1}\big|\big\langle\psi, b_x^{*}b_y\psi\big\rangle\big|\leq \kappa^{(n+1)}   \big\langle\psi,\,N\psi\big\rangle.
\end{align}
Applying inequalities \eqref{bound first order},  \eqref{bound higher order}, and \eqref{bound remainder} to \eqref{bound on comm} we obtain
\begin{align*}
     \Big|\Big\langle\psi, i[H,N_{f,ts}]\psi\Big\rangle\Big|\leq &\frac{\kappa}{s}\Big\langle\psi,\,N_{f'_{ts}}\psi\Big\rangle\\
     &+C\sum_{l=1}^{n-1}\sum_{m=l}^n \sum_{\und k\in \cK_{l,m}} \;\sum_{\und j\in \cJ_{\und k}}\frac{4^{m-l}\cdot s^{-m-1}}{(k_1-1)!\dots (k_l-1)! }\kappa^{(m)}\Big\langle\psi,\,N_{f'_{J_l,ts}}\psi\Big\rangle\nonumber\\
     &+C_f4^n\Big(\frac{1}{n!}+\sum_{l=2}^{n-1}\left(\frac{e}{2}\right)^l\sum_{m=l-1}^n  2^m\Big)s^{-n-1}\kappa^{(n+1)}   \big\langle\psi,\,N\psi\big\rangle.\nonumber
\end{align*}

Together with equations \eqref{derivative exp value N} and \eqref{deriv astlo} we obtain 

\begin{align}
    \frac{d}{dt}&\langle\: N_{f,ts}\rangle_t\nonumber\\
    \leq& \frac{\kappa-v}{s}\langle\: N_{f',ts}\rangle_t+C\sum_{l=1}^{n-1}\sum_{m=l}^n \sum_{\und k\in \cK_{l,m}} \;\sum_{\und j\in \cJ_{\und k}}\frac{4^{m-l}\cdot s^{-m-1}}{(k_1-1)!\dots (k_l-1)! }\kappa^{(m)}\big\langle N_{f'_{J_l,ts}}\big\rangle_t\nonumber\\
    &+C_f4^n\Big(\frac{1}{n!}+\sum_{l=2}^{n-1}\left(\frac{e}{2}\right)^l\sum_{m=l-1}^n  2^m\Big)s^{-n-1}\kappa^{(n+1)}   \big\langle N\big\rangle_t\nonumber\\
    =&\frac{\kappa-v}{s}\langle\: N_{f',ts}\rangle_t+C\sum_{l=1}^{n-1}\sum_{m=l}^n \sum_{\und k\in \cK_{l,m}} \;\sum_{\und j\in \cJ_{\und k}}\frac{4^{m-l}\cdot s^{-m-1}}{(k_1-1)!\dots (k_l-1)! }\kappa^{(m)}\big\langle N_{f'_{J_l,ts}}\big\rangle_t\label{before integrating}\\
    &+C_f4^n\Big(\frac{1}{n!}+\sum_{l=2}^{n-1}\left(\frac{e}{2}\right)^l\sum_{m=l-1}^n  2^m\Big)s^{-n-1}\kappa^{(n+1)}   \big\langle N\big\rangle_0,\nonumber
\end{align}
where the equality is due to the fact that $N$ commutes with the Hamiltonian. 
We conclude the proof by integrating both sides of equation \eqref{before integrating} from 0 to $t$.
\end{proof}

Iterating  \thmref{theorem bound on integral of exp of N} and bounding the sums obtained through this process we can derive a bound that only contains averages at the initial time of the number operator and of the ASTLO. For any positive integers $p$ and $n$ we can define the following sets
\begin{align*}
    &\cL_{p}=\{\und l=(l_1,\dots l_p)\in \N^p\text{ such that } 1\le l_i\le n-1- \sum_{j=1}^{i-1} l_j\quad\forall i\in[1,p]\},\\
    &\cM_{\und l, p}=\{\und m=(m_1,\dots,m_p)\in\N^p\text{ such that }l_i\le\ m_i\le n-\sum_{j=1}^{i-1}m_j\quad\forall i\in[1,p]\}
\end{align*}
For vectors $\und l\in\cL_{p}$ and $\und m\in\cM_{\und l,p}$, we can define the following quantities

\begin{align*}
   L_p=\sum_{i=1}^p l_i\qquad M_p=\sum_{i=1}^p m_i.
\end{align*}

\begin{corollary}[Iterated bound]\label{cor: iter bound}
For any $R>0$, range of the interactions, for all $f\in\cF$, $ n\geq1$, and $c>\kappa$ there exist two positive constants $C_{f,c}, C_{R}$ such that, for all $s>C_R$ and $t,\rho>0$, the following holds
\begin{align*}
    \int_0^t &\langle\: N_{f',rs}\rangle_r dr\\
    &\leq 4\kappa^{(0)}C_{f,c}t\left(\frac{C_{f,c}R }{s}\right)^n \\
    &\quad +C_{f,c}s\langle N_{f,ts}\rangle_0\nonumber\\
    &\quad +C_{f,c}\kappa^{(0)}R  t \big\langle N\big\rangle_0\Big(\left(\frac{4R}{s}\right)^n\frac{1}{n!}+\left(\frac{4eR}{s}\right)^{n}\Big)\nonumber\\
    &\quad +(\kappa^{(0)})^2 Rt\left(\frac{4R}{s}\right)^n\big\langle N\big\rangle_0\sum_{a=1}^nC_{f,c}^a\left(1-\frac{R}{s}\right)^{-a}\left(\frac{4s}{3eR}-1\right)^{-a}\nonumber\\
    &\quad +(\kappa^{(0)})^2R t\left(\frac{8R}{s}\right)^n \big\langle N\big\rangle_0\sum_{a=1}^n C_{f,c}^a\left(1-\frac{R}{2s}\right)^{-a}\left(\frac{8s}{3eR}-1\right)^{-a}\nonumber\\
    &\quad +\kappa^{(0)}s\sum_{a=0}^{n-1} C_{f,c}^{a+1}\sum_{\und l\in\cL_{a}}4^{-L_{a}}\sum_{\und m\in\cM_{\und l,a}}\left( \frac{4R}{s}\right)^{M_{a}}\\
    &\qquad\times\sum_{t=1}^{a}\;\sum_{\und k^{(t)}\in\cK_{l_t,m_t}} \;\sum_{\und j^{(t)}\in\cJ_{\und k^{(t)}}}\frac{1}{(k_1^{(1)}-1)!\dots (k_{l_{a}}^{(a)}-1)! }\langle N_{f_{J_{l_1}+..J_{l_{a}}},ts}\rangle_0 ,\nonumber
\end{align*}  
  where for $n=1$ the sums should be dropped.
\end{corollary}

\begin{proof}

Since, under these assumptions,  \thmref{theorem bound on integral of exp of N} holds and $\kappa - v>0$, after dropping $\langle N_{f,ts}\rangle_t$ and multiplying by $s(v-\kappa)^{-1}$ in inequality \eqref{lemma before recursion}, it follows that
    
\begin{align}
    \int_0^t \langle\: N_{f',rs}\rangle_r dr\leq C_{f,c}&\Bigg(\sum_{l=1}^{n-1}\sum_{m=l}^n \;\sum_{\und k\in \cK_{l,m}} \;\sum_{\und j\in \cJ_{\und k}}\frac{C\cdot 4^{m-l}\cdot s^{-m}}{(k_1-1)!\dots (k_l-1)! }\kappa^{(m)}\int_0^t\big\langle N_{f'_{J_l,rs}}\big\rangle_rdr \\
    &+ s\langle N_{f,ts}\rangle_0+4^n\Big(\frac{1}{n!}+\sum_{l=2}^{n-1}\left(\frac{e}{2}\right)^l\sum_{m=l-1}^n  2^m\Big)s^{-n}\kappa^{(n+1)}  t \big\langle N\big\rangle_0\Bigg).\nonumber
\end{align}
 We notice that this inequality holds for any $n\in\mathbb{N}$ and any $f\in\mathcal{E}$. This allows us to iterate it on the term $\int_0^t\big\langle N_{f'_{J,rs}}\big\rangle dr$ with $n-m_1$ instead that with $n$:
 \begin{align*}
    &\int_0^t \langle\: N_{f',rs}\rangle_r dr\\
    &\leq  C_{f,c}^2\sum_{\und l\in \cL_2}\;\sum_{\und m\in\cM_{\und l,2}} \; \sum_{t=1,2}\;\sum_{\und k^{(t)}\in \cK_{l_t,m_t}} \;\sum_{\und j^{(t)}\in \cJ_{\und k^{(t)}}}\frac{4^{{m_1}+m_2-{l_1}-l_2}\cdot s^{-{m_1}-m_2}}{(k_1^{(1)}-1)!\dots (k_{l_2}^{(2)}-1)! }\\
    &\qquad \times\kappa^{({m_1})}\kappa^{({m_2})}\int_0^t\big\langle N_{f'_{J_{l_1}+J_{l_2},rs}}\big\rangle_rdr\nonumber \\
    &+  C_{f,c}\sum_{{l_1}=1}^{n-1}\sum_{{m_1}={l_1}}^n \sum_{\und k^{(1)}\in \cK_{l_1,m_1}} \;\sum_{\und j^{(1)}\in \cJ_{\und k^{(1)}}}\frac{ 4^{{m_1}-l_1}\cdot s^{-{m_1}}}{(k_1^{(1)}-1)!\dots (k_{l_1}^{(1)}-1)! }\kappa^{({m_1})}\nonumber\\
    &\qquad\times\left( s\langle N_{f_{J_{l_1}},ts}\rangle_0+4^{n-m_1}\Big(\frac{1}{(n-m_1)!}+\sum_{{l_2}=2}^{n-1-l_1}\left(\frac{e}{2}\right)^{l_2}\sum_{{m_2}={l_2}-1}^{n-m_1}  2^{m_2}\Big)s^{-n}\kappa^{(n+1)}  t \big\langle N\big\rangle_0\right)\nonumber\\
    &+ C_{f,c}s\langle N_{f,ts}\rangle_0+C_{f,c}4^n\Big(\frac{1}{n!}+\sum_{{l_1}=2}^{n-1}\left(\frac{e}{2}\right)^{l_1}\sum_{{m_1}={l_1}-1}^n  2^{m_1}\Big)s^{-n}\kappa^{(n+1)}  t \big\langle N\big\rangle_0.\nonumber
\end{align*}
We used the fact that $m_1\geq l_1$ to bound
\begin{align*}
    \sum_{l_2=1}^{n-1-m_1}1\leq \sum_{l_2=1}^{n-1-l_1}1.
\end{align*}
  
 Now we iterate this procedure $n$ times, for $n-m_1\dots-m_a$ at the a-th step of the iteration, then $l_1+..+l_n=n-1$, $m_1+...+m_n=n$. This way, we obtain
\begin{align}\label{after n it}
  &\int_0^t \langle\: N_{f',rs}\rangle_r dr\\
    &\leq  C_{f,c}^n 4 \kappa^{(0)}s^{-n}R^n\int_0^t\big\langle N_{f'_n,rs}\big\rangle_rdr\nonumber \\
    &+ \sum_{a=0}^{n-1} C_{f,c}^{a+1}\sum_{\und l\in\cL_{a}}\sum_{\und m\in\cM_{\und l,a}}\sum_{t=1}^{a}\;\sum_{\und k^{(t)}\in\cK_{l_t,m_t}} \;\sum_{\und j^{(t)}\in\cJ_{\und k^{(t)}}}\frac{ 4^{{M_{a}}-L_{a}}\cdot s^{-{M_{a}}}}{(k_1^{(1)}-1)!\dots (k_{l_{a}}^{(a)}-1)! }\kappa^{(0)}R^{M_{a}}\nonumber\\
    &\quad\times\left( s\langle N_{f_{J_{l_1}+..J_{l_{a}}},ts}\rangle_0+4^{n-M_{a}}\Big(\frac{1}{(n-M_{a})!}+\sum_{{l_{a+1}}=2}^{n-1-L_{a}}\left(\frac{e}{2}\right)^{l_{a+1}}\sum_{{m_{a+1}}={l_{a+1}}}^{n-M_{a}}  2^{m_{a+1}}\Big)s^{-n}\kappa^{(0)}R^{n+1}  t \big\langle N\big\rangle_0\right)\nonumber\\
    &+ C_{f,c}s\langle N_{f,ts}\rangle_0+C_{f,c}4^n\Big(\frac{1}{n!}+\sum_{{l_1}=2}^{n-1}\left(\frac{e}{2}\right)^{l_1}\sum_{{m_1}={l_1}-1}^n  2^{m_1}\Big)s^{-n}\kappa^{(0)}R^{n+1}  t \big\langle N\big\rangle_0.  \nonumber
\end{align}
  where we used the fact that $\kappa^{(p)}\leq \kappa^{(0)}R^p$.
Now we want to track the $n$ dependency in the expression in order to understand its behaviour as $n\to\infty$. We start with the first summand
  \begin{align*}
   C_{f,c}^n 4 \kappa^{(0)}s^{-n}R^n\int_0^t\big\langle N_{f'_n,rs}\big\rangle_rdr\leq 4\kappa^{(0)}\left(\frac{C_{f,c}R }{s}\right)^nC_{f}t, 
\end{align*}
  where we used the result from  \lemref{Upper bounds on the derivatives}.
We will use the following elementary estimate on the geometric sums appearing in \eqref{after n it}
 \begin{align}\label{geometric series}
     \sum_{n=a}^b c^n\leq \begin{cases}
     c^a(1-c)^{-1},\qquad &0<c<1\\
     c^b (c-1)^{-1},\qquad &c>1.
    \end{cases}
 \end{align}
 
 We consider the remaining terms in \eqref{after n it}. We start with the remainder term

    \begin{align*}
    C_{f,c}4^n&\Big(\frac{1}{n!}+\sum_{{l_1}=2}^{n-1}\left(\frac{e}{2}\right)^{l_1}\sum_{{m_1}={l_1}-1}^n  2^{m_1}\Big)s^{-n}\kappa^{(0)}R^{n+1}  t \big\langle N\big\rangle_0\\
    \leq &C_{f,c}\left(\frac{4R}{s}\right)^n\Big(\frac{1}{n!}+2^n\sum_{{l_1}=2}^{n-1}\left(\frac{e}{2}\right)^{l_1}\Big)\kappa^{(0)}R  t \big\langle N\big\rangle_0\\
    \leq &C_{f,c}\left(\frac{4R}{s}\right)^n\Big(\frac{1}{n!}+C2^n\left(\frac{e}{2}\right)^{n}\Big)\kappa^{(0)}R  t \big\langle N\big\rangle_0\\
     = &C_{f,c}\Big(\left(\frac{4R}{s}\right)^n\frac{1}{n!}+C\left(\frac{4eR}{s}\right)^{n}\Big)\kappa^{(0)}R  t \big\langle N\big\rangle_0,
\end{align*}
where $C=\left(\frac{e}{2}-1\right)^{-1}$. Now we consider the term
\begin{align}\label{first big term after n}
    \sum_{a=0}^{n-1} &C_{f,c}^{a+1}\sum_{\und l\in\cL_{a}}\sum_{\und m\in\cM_{\und l,a}}\sum_{t=1}^{a}\;\sum_{\und k^{(t)}\in\cK_{l_t,m_t}} \;\sum_{\und j^{(t)}\in\cJ_{\und k^{(t)}}}\frac{ 4^{{M_{a}}-L_{a}} s^{-{M_{a}}}}{(k_1^{(1)}-1)!\dots (k_{l_{a}}^{(a)}-1)! }\kappa^{(0)}R^{M_{a}}\nonumber\\
    &\quad\times\frac{4^{n-M_{a}}}{(n-M_{a})!}s^{-n}\kappa^{(0)}R^{n+1}  t \big\langle N\big\rangle_0\nonumber\\
    &=\left(\frac{4R}{s}\right)^n(\kappa^{(0)})^2 Rt\big\langle N\big\rangle_0\sum_{a=0}^{n-1} C_{f,c}^{a+1}\sum_{\und l\in\cL_{a}}\sum_{\und m\in\cM_{\und l,a}}\frac{ 4^{{M_{a}}-L_{a}} R^{M_{a}}}{(4s)^{M_{a}}(n-M_{a})!}\\
    &\quad\times\sum_{t=1}^{a}\;\sum_{\und k^{(t)}\in\cK_{l_t,m_t}} \;\sum_{\und j^{(t)}\in\cJ_{\und k^{(t)}}}\frac{ 1}{(k_1^{(1)}-1)!\dots (k_{l_{a-1}}^{(a-1)}-1)! }.\nonumber
    \end{align}
 We focus first on the sums over the $k$'s and the $j$'s and we find
\begin{align}
    &\sum_{\und k\in\cK_{l,m}} \;\sum_{\und j\in\cJ_{\und k}}\frac{1}{(k_1-1)!\dots (k_l-1)! }\nonumber\\
    &=\sum_{\und k\in\cK_{l,m}}\frac{(k_1+1)\cdot\dots\cdot (k_l+1)}{(k_1-1)!\dots (k_l-1)! }\nonumber\\
   &\leq\sum\limits_{\substack{1\leq k_i\leq m,\\ \forall 1\leq i\leq l}} \frac{(k_1+1)\cdot\dots\cdot (k_l+1)}{(k_1-1)!... (k_l-1)! } \nonumber\\
   &=\left(\sum_{k=1}^m\frac{k+1}{(k-1)!}\right)^l\nonumber \\
   &=\left(\sum_{k=1}^m\frac{k}{(k-1)!}+\sum_{k=1}^m\frac{1}{(k-1)!}\right)^l\nonumber \\
   &\leq \left(3e\right)^l. \label{simplified sum}
\end{align}
 Applying \eqref{simplified sum} we can bound \eqref{first big term after n} by

 \begin{align*}
      \left(\frac{4R}{s}\right)^n&(\kappa^{(0)})^2 Rt\big\langle N\big\rangle_0\sum_{a=0}^{n-1} C_{f,c}^{a+1}\sum_{\und l\in\cL_{a}}\sum_{\und m\in\cM_{\und l,a}}\frac{ 4^{{M_{a}}} R^{M_{a}}}{(4s)^{M_{a}}(n-M_{a})!}\left(\frac{3e}{4}\right)^{L_{a}}\\
    % &=\left(\frac{4R}{s}\right)^n(\kappa^{(0)})^2 Rt\big\langle N\big\rangle_0\sum_{a=1}^nC_{f,c}^a\sum_{\substack{1\leq l_1\leq n-1\\\dots \\ 1 \leq l_{a-1} \leq n-1-L_{a-2}}}\left(\frac{e}{2}\right)^{L_{a-1}}\sum_{\substack{l_1\leq m_1\leq n\\\dots\\ l_{a-1}\leq m_{a-1}\leq n-M_{a-2}}}\left(\frac{R}{s}\right)^{M_{a-1}} \frac{ 1}{(n-M_{a-1})!}\\
    &\leq \left(\frac{4R}{s}\right)^n(\kappa^{(0)})^2 Rt\big\langle N\big\rangle_0\sum_{a=1}^nC_{f,c}^{a}\left(1-\frac{R}{s}\right)^a\sum_{\und l\in\cL_{a-1}}\left(\frac{3eR}{4s}\right)^{L_{a-1}}\\
     &\leq \left(\frac{4R}{s}\right)^n(\kappa^{(0)})^2 Rt\big\langle N\big\rangle_0\sum_{a=1}^nC_{f,c}^a\left(1-\frac{R}{s}\right)^{-a}\left(1-\frac{3eR}{4s}\right)^{-a}\left(\frac{3eR}{4s}\right)^{a}\\
     &= \left(\frac{4R}{s}\right)^n(\kappa^{(0)})^2 Rt\big\langle N\big\rangle_0\sum_{a=1}^nC_{f,c}^a\left(1-\frac{R}{s}\right)^{-a}\left(\frac{4s}{3eR}-1\right)^{-a},
 \end{align*}\\
where we assumed $4s>3eR$.
In the same way we can bound the following term

\begin{align*}
    \sum_{a=0}^{n-1} &C_{f,c}^{a+1}\sum_{\und l\in\cL_{a}}\sum_{\und m\in\cM_{\und l,a}}\sum_{t=1}^{a}\;\sum_{\und k^{(t)}\in\cK_{l_t,m_t}} \;\sum_{\und j^{(t)}\in\cJ_{\und k^{(t)}}}\frac{ 4^{{M_{a}}-L_{a}}\cdot s^{-{M_{a}}}}{(k_1^{(1)}-1)!\dots (k_{l_{a}}^{(a)}-1)! }\kappa^{(0)}R^{M_{a}}\\
    &\quad\times4^{n-M_{a}}\sum_{{m_{a+1}}={l_{a+1}}-1}^{n-M_{a}}  2^{m_{a+1}}s^{-n}\kappa^{(0)}R^{n+1}  t \big\langle N\big\rangle_0\nonumber\\
    % &=(\kappa^{(0)})^2R^{n+1}4^n s^{-n} t \big\langle N\big\rangle_0\sum_{a=1}^n C_{f,c}^a\sum_{\substack{1\leq l_1\leq n-1\\\dots \\ 1 \leq l_{a-1} \leq n-1-L_{a-2}}}\sum_{\substack{l_1\leq m_1\leq n\\\dots\\ l_{a}\leq m_{a}\leq n-M_{a-1}}}\frac{R^{M_{a-1}}4^{{M_{a-1}}-L_{a-1}}}{ 4^{M_{a-1}}s^{{M_{a-1}}}}2^{m_a}\\
    % &\quad\sum_{\substack{1\leq k_i^{(t)},\,\forall i\in[1,{l_t}]\\k_1^{(t)}+\dots+k_{l_t}^{(t)}={m_t}\\ t=1,\dots a-1}} \sum_{\substack{0\leq j_i^{(t)}\leq k_i^{(t)} \\ \forall i\in[1,{l_t}]\\ t=1,\dots a-1}}\frac{ 1}{(k_1^{(1)}-1)!\dots (k_{l_{a-1}}^{(a-1)}-1)! }\\
  &\leq(\kappa^{(0)})^2R\left(\frac{4R}{s}\right)^n t \big\langle N\big\rangle_0\sum_{a=0}^{n-1} C_{f,c}^{a+1}\sum_{\und l\in\cL_{a}}\sum_{\und m\in\cM_{\und l,a}} \left(\frac{R}{2s}\right)^{M_{a}}\left(\frac{3e}{4}\right)^{L_{a}}2^n\\
  &\leq(\kappa^{(0)})^2R t\left(\frac{8R}{s}\right)^n \big\langle N\big\rangle_0\sum_{a=0}^{n-1} C_{f,c}^{a+1}\left(1-\frac{R}{2s}\right)^{-a-1}\sum_{\und l\in\cL_{a}}\left(\frac{3eR}{8s}\right)^{L_{a}}\\
   &\leq(\kappa^{(0)})^2R t\left(\frac{8R}{s}\right)^n \big\langle N\big\rangle_0\sum_{a=1}^n C_{f,c}^a\left(1-\frac{R}{2s}\right)^{-a}\left(\frac{8s}{3eR}-1\right)^{-a},
\end{align*}
and this concludes the proof.

\end{proof}

\subsection{Relation Between the Number Operator and the ASTLO}\label{main res}

In this chapter we derive a relation between the ASTLOs and the number operator. Applying these results to those obtained in the previous chapter, we derive a bound on the evolved ASTLO that depends only on the ASTLO and the number operator at initial times. Then, by understanding the decay of the terms in this expression as $n$ goes to infinity, we obtain a bound that only involves the average at initial time of the number operator on a given finite set. Together with the result obtained at the beginning of the chapter, this allows us to derive the bound on the particle propagation.

\subsubsection{Number Operator and ASTLO}

We now compare the ASTLO and the number operator on certain finite regions. The specific properties of the function in $\cF$ allows us to derive the following bounds.
\begin{theorem}[Connection between number operator and ASTLOs]\label{N and ASLTO}
Assuming that the initial state has controlled density, so that there exists a constant $\delta>0$ such that $\delta ^{-1}\le\omega(n_x)\le\delta$ for every $x\in\Lambda$, for any $c>\kappa$, $s>0$, and $v'>v$, the following holds \\
1) There exists a function $f\in \cF$ such that  
\begin{equation}
    \big\langle N_{|x|<b}\big\rangle_t < 2 \big\langle N_{f_{ts}}\big\rangle_t \label{ineq N Nf at time t}
\end{equation}
2) For any  function $f\in\mathcal{E}$ there exists a constant $C_{f,d}$ such that 
        \begin{equation}
        \big\langle N_{f_{a,0s}}\big\rangle_0\le C_{f,d}a^{3d+1}\delta^2\left(\frac{s^d}{\rho^d}+1\right)\big\langle N_{|x|<\rho}\big\rangle_0 \qquad \text{for any }a\geq 0,\label{ineq N Nf at time 0} 
    \end{equation}

  for any $b>0$ and $\rho=v't+b$, where $f_a\in\mathcal{E} $ is such that $f'_a=\left(\check\varphi^{(a)}\right)^4$ where $\varphi\in C_c^{\infty}(B_1(0))$ is such that $f'=\check\varphi^4$.
\end{theorem}

\begin{proof}
    We start by showing \eqref{ineq N Nf at time t}, we want to prove that there exists a function $f\in\cF$ such that
    \begin{equation*}
        \sum_{|x|<b}\langle n_x\rangle_t <2 \sum_{x\in\Lambda}f\left(\frac{\rho-vt-|x|}{s}\right)\langle n_x\rangle_t.
    \end{equation*}
    This holds if 
    \begin{equation*}
        f\left(\frac{\rho-vt-|x|}{s}\right)>\frac{1}{2}\qquad \text{for}\;|x|\leq b.
    \end{equation*}
    Given $\rho=v't+b$ with $v'>v$ and by the monotonicity of the function $f$, the last equation holds if
    \begin{equation}\label{condition on f 2}
        f\left(\frac{b-|x|}{s}\right)>\frac{1}{2}\qquad \text{for}\;|x|\leq b.
    \end{equation}
    If $f(0)\geq 1/2$, condition \eqref{condition on f 2} holds. We observe that we can always find a function in $\cF$ that respects this condition, just by shifting a given function in the class.\\
    Now we show \eqref{ineq N Nf at time 0}. Since we assume that the initial state has controlled density it holds
    \begin{equation*}
        \big\langle N_{|x|<\rho}\big\rangle_0 \ge\delta^{-1}\sum_{|x|<\rho}1=\delta^{-1}|B_{\rho}|,
    \end{equation*}
 for some positive constant $\rho$. So, what we want to show is that for any function $f\in\cF$ there exist a constant $C\equiv C({f,d,\rho,s,\delta})>0$ such that
    \begin{align*}
        &\sum_{x\in\Lambda}f_a\left(\frac{\rho-|x|}{s}\right)\langle n_x\rangle_0\le \delta^{-1} Ca^{3d+1}|B_{\rho}|\\
        \Longleftarrow
        %\qquad & \delta\sum_{x\in\Lambda}f_a\left(\frac{\rho-|x|}{s}\right)\le \delta^{-1} C|B_{\rho}|\\
 \qquad & \sum_{x\in\Lambda}f_a\left(\frac{\rho-|x|}{s}\right)\le \delta^{-2}Ca^{3d+1}|B_{\rho}|,
    \end{align*}
    for every integer $a\geq0$. 
    We start by splitting the sum in the following way
    \begin{align}
         \sum_{x\in\Lambda}f_a\left(\frac{\rho-|x|}{s}\right)&= \sum_{|x|\leq \rho}f_a\left(\frac{\rho-|x|}{s}\right)+\sum_{|x|\geq \rho}f_a\left(\frac{\rho-|x|}{s}\right)\nonumber\\
         &\leq C_f\sum_{|x|\leq \rho}1+\sum_{|x|\geq \rho}f_a\left(\frac{\rho-|x|}{s}\right)\nonumber\\
         &=C_{f,d}|B_{\rho}|+\sum_{|x|\geq \rho}f_a\left(\frac{\rho-|x|}{s}\right),\label{eq for Nf at 0}
    \end{align}
    where we used the fact that $f_a\leq C_f$ for every $a$ by  \lemref{Upper bounds on the derivatives}. The first term of \eqref{eq for Nf at 0} is exactly as we wanted, so we can focus on the second. We bound it by the integral associated that we write in polar coordinates
    \begin{align}
        \sum_{|x|\geq \rho}f_a\left(\frac{\rho-|x|}{s}\right)&\leq |S_{d-1}|\int_{\rho}^{\infty}f_a\left(\frac{\rho-z}{s}\right)z^{d-1}dz\nonumber\\
%        &=|S_{d-1}|\int_{-\infty}^{0}f_a\left(\frac{r}{s}\right)\left(\rho-r\right)^{d-1}dr \nonumber\\
        &=|S_{d-1}|s^d\int_{-\infty}^{0}f_a\left(v\right)(\frac{\rho}{s}-v)^{d-1}dv\nonumber\\
        &=|S_{d-1}|\frac{s^d}{d}\Big(\int_{-\infty}^{0}f'_a\left(v\right)(\frac{\rho}{s}-v)^{d}dv+\Big[f_a(v)(\frac{\rho}{s}-v)^{d}\Big]^0_{\infty}\Big),\label{int plus boundary}
    \end{align}

    where we obtained the first two equality by substitution and the last one integrating by parts.  We notice that, since $f'_a=\left(\check\varphi^{(a)}\right)^4$ and $\check\varphi^{(a)}$ is a Schwartz function, decaying faster than any polynomial, the boundary term becomes

    \begin{equation*}
        |S_{d-1}|\frac{s^d}{d}\Big[f_a(v)(\frac{\rho}{s}-v)^{d}\Big]^0_{\infty}=|S_{d-1}|\frac{s^d}{d}f_a(0)\left(\frac{\rho}{s}\right)^{d}\le C_{f,d} |B_{\rho}|
    \end{equation*}

    So we can focus on the integral in \eqref{int plus boundary}
    \begin{align*}
        |S_{d-1}|&\frac{s^d}{d}\int_{-\infty}^{0}f'_a\left(v\right)(\frac{\rho}{s}-v)^{d}dv\\
        &=-|S_{d-1}|\frac{s^d}{d}\int_{-\infty}^{0}\left(\check\varphi^{(a)}\right)^4(\frac{\rho}{s}-v)^{d}dv\\
        &\leq -|S_{d-1}|\frac{s^d}{d}||\check\varphi^{(a)}||^3_{\infty}\int_{-\infty}^{0}\abs{\check\varphi^{(a)}}(\frac{\rho}{s}-v)^{d}dv\\
        &= -|S_{d-1}|\frac{s^d}{d}||\check\varphi^{(a)}||^3_{\infty}\int_{-\infty}^{0}\abs{\check\varphi^{(a)}}(\frac{\rho}{s}-v)^{d}\frac{(1+v^2)^d}{(1+v^2)^d}dv\\
    &\leq -|S_{d-1}|\frac{s^d}{d}||\check\varphi^{(a)}||^3_{\infty}||\check\varphi^{(a)}(\frac{\rho}{s}-v)^{d}(1+v^2)^d||_{\infty}\int_{-\infty}^{0}\frac{1}{(1+v^2)^d}dv.
    \end{align*}
    Using the fact that $\varphi\in C_c^{\infty}(B_1(0))$, as well as its derivatives, we obtain
    
        \begin{equation*}
            ||\check\varphi^{(a)}||^3_{\infty}\leq C||\xi^a\varphi||_1^3=C||\varphi||_1^3=C_f.
        \end{equation*}
    Now we notice that    
       \begin{align*}
           ||\check\varphi^{(a)}&v^l||_{\infty}\\
           &\leq C||\partial_{\xi}^l\left(\xi^a\varphi\right)||_1\\
           &\leq C\sum_{i=0}^l \binom{l}{i}||\partial^{(i)}_{\xi}\xi^a\cdot\partial^{(l-i)}_{\xi}\varphi||_1\\
           &= C\sum_{i=0}^l \binom{l}{i}||\frac{a!}{(a-i-1)!}\xi^{a-i}\cdot\partial^{(l-i)}_{\xi}\varphi||_1\\
           &\leq Ca^{l+1}\sum_{i=0}^l \binom{l}{i}||\partial^{(l-i)}_{\xi}\varphi||_1\\
           &\leq C_{l,f}a^{l+1}.
       \end{align*}
  This way, by linearity, 
  \begin{align*}
      ||\check\varphi^{(a)}(\frac{\rho}{s}-v)^{d}(1+v^2)^d||_{\infty}\leq C_{d,f}a^{3d+1}.
  \end{align*}
 And since $\int_{-\infty}^{0}\frac{1}{(1+v^2)^d}dr\leq C_{d}$ we can obtain the following bound
 \begin{align*}
      \sum_{|x|\geq \rho}f_a\left(\frac{\rho-|x|}{s}\right)\leq C_{f,d}|S_{d-1}|\frac{s^d}{d}a^{3d+1}.
 \end{align*}
 This implies

 \begin{align*}
     \sum_{x\in\Lambda}f_a\left(\frac{\rho-|x|}{s}\right)&\leq \sum_{|x|\geq \rho}f_a\left(\frac{\rho-|x|}{s}\right)+\sum_{|x|\leq \rho}f_a\left(\frac{\rho-|x|}{s}\right)\\
     &\leq C_{f,d}|S_{d-1}|\frac{s^da^{3d+1}}{d}+C_{f,d}|B_{\rho}|\\
     &\le C_{f,d}a^{3d+1}\left(\frac{s^d}{\rho^d}+1\right)|B_{\rho}|.
 \end{align*}

 This way 
 \begin{align*}
     \langle N_{f_{a,0s}}\rangle_0\le C_{f,d}a^{3d+1}\delta\left(\frac{s^d}{\rho^d}+1\right)|B_{\rho}|\le C_{f,d}a^{3d+1}\delta^2\left(\frac{s^d}{\rho^d}+1\right) \langle N_{|x|\le\rho}\rangle_0
 \end{align*}
\end{proof}
Noticing that condition \eqref{ineq N Nf at time 0} holds for any function in our class we can bound the integral of the ASTLO by an expression that only contains expectation at initial time of the number operator on the whole lattice and restricted to a finite region.

\begin{corollary}[Bound in terms of number operator on finite region]\label{cor bound on integral before limit}
 Consider an initial state such that there exists a constant $\delta>0$ such that $\delta ^{-1}\le\omega(n_x)\le\delta$ for every $x\in\Lambda$. For any $R>0$, range of the interactions, there exists a constant $C_{d,R}$ such that for all $s>C_{d,R}$, $c>\kappa$, $v'>v$, $t>0$, and $f\in\cF$   it holds 

 \begin{align}
     \int_0^t &\langle\: N_{f',rs}\rangle_r dr\\
    &\leq \kappa^{(0)}C_{f,c}t\left(\frac{C_{f,c}R }{s}\right)^n \nonumber\\
   &\quad + C_{f,d}\delta^2\left(\frac{s^d}{\rho^d}+1\right)\big\langle N_{|x|<\rho}\big\rangle_0\nonumber\\
    &\quad +\kappa^{(0)}R  tC_{f,c} \big\langle N\big\rangle_0\Big(\left(\frac{4R}{s}\right)^n\frac{1}{n!}+\left(\frac{4eR}{s}\right)^{n}\Big)\nonumber\\
    &\quad +(\kappa^{(0)})^2 Rt\left(\frac{4R}{s}\right)^n\big\langle N\big\rangle_0\sum_{a=1}^nC_{f,c}^a\left(1-\frac{R}{s}\right)^{-a}\left(\frac{4s}{3eR}-1\right)^{-a}\nonumber\\
    &\quad +(\kappa^{(0)})^2R t\left(\frac{8R}{s}\right)^n \big\langle N\big\rangle_0\sum_{a=1}^n C_{f,c}^a\left(1-\frac{R}{2s}\right)^{-a}\left(\frac{8s}{3eR}-1\right)^{-a}\nonumber\\
    &\quad +\kappa^{(0)}C_{f,d}\delta^2\left(\frac{s^d}{\rho^d}+1\right)\big\langle N_{|x|<\rho}\big\rangle_0\sum_{a=1}^na^{3d+1} C_{f,c}^a\left(1- \frac{36edR}{s}\right)^{-a+1}\left(\frac{s}{9edR}-1\right)^{-a+1},\label{inequality before limit}
\end{align}
for any $b>0$ and $\rho=v't+b$.
\end{corollary}

\begin{proof}
    By applying  \thmref{N and ASLTO} to the last summand of  \corref{cor: iter bound}
 we can bound it in the following way
 
 \begin{align*}
   \kappa^{(0)}s&\sum_{a=0}^{n-1} C_{f,c}^{a+1}\sum_{\und l\in\cL_{a}}4^{-L_{a}}\sum_{\und m\in\cM_{\und l,a}}\left( \frac{4R}{s}\right)^{M_{a}}\\
    &\qquad\times\sum_{t=1}^{a}\;\sum_{\und k^{(t)}\in\cK_{l_t,m_t}} \;\sum_{\und j^{(t)}\in\cJ_{\und k^{(t)}}}\frac{1}{(k_1^{(1)}-1)!\dots (k_{l_{a}}^{(a)}-1)! }\langle N_{f_{J_{l_1}+..J_{l_{a}}},ts}\rangle_0 ,\nonumber\\
    &\leq  \kappa^{(0)}s\sum_{a=0}^{n-1} (a+1)^{3d}C_{f,c}^{a+1}\sum_{\und l\in\cL_{a}}4^{-L_{a}}\sum_{\und m\in\cM_{\und l,a}}\left( \frac{4R}{s}\right)^{M_{a}}\\\label{very ugly}
    &\qquad\times\sum_{t=1}^{a}\;\sum_{\und k^{(t)}\in\cK_{l_t,m_t}} \;\sum_{\und j^{(t)}\in\cJ_{\und k^{(t)}}}\frac{(J_{l_1}+...+J_{l_{a}})^{3d}}{(k_1^{(1)}-1)!\dots (k_{l_{a}}^{(a)}-1)! } C_{f,d}\delta^2\left(\frac{s^d}{\rho^d}+1\right)  \big\langle N_{|x|<\rho}\big\rangle_0.\nonumber
    \end{align*}
Since $J_i\leq k_i$ for every $i\in[1,l_t]$ and the $k_i$'s sum up to $m_t$ we derive the following
\begin{align*}
    &\sum_{t=1}^{a}\;\sum_{\und k^{(t)}\in\cK_{l_t,m_t}} \;\sum_{\und j^{(t)}\in\cJ_{\und k^{(t)}}}\frac{(J_{l_1}+...+J_{l_{a}})^{3d}}{(k_1^{(1)}-1)!\dots (k_{l_{a}}^{(a)}-1)! }\\
    &\quad \leq \sum_{t=1}^{a}\;\sum_{\und k^{(t)}\in\cK_{l_t,m_t}} \;\sum_{\und j^{(t)}\in\cJ_{\und k^{(t)}}}\frac{ (m_{1}+...+m_{a})^{3d}}{(k_1^{(1)}-1)!\dots (k_{l_{a}}^{(a)}-1)! }\\
    &\quad = {M_{a}}^{3d}\sum_{t=1}^{a}\;\sum_{\und k^{(t)}\in\cK_{l_t,m_t}} \;\sum_{\und j^{(t)}\in\cJ_{\und k^{(t)}}}\frac{1}{(k_1^{(1)}-1)!\dots (k_{l_{a}}^{(a)}-1)! }.
\end{align*}
This way we can bound \eqref{very ugly} by

    \begin{align}
   \kappa^{(0)}&C_{f,d}\delta^2\left(\frac{s^d}{\rho^d}+1\right)\big\langle N_{|x|<\rho}\big\rangle_0\sum_{a=0}^{n-1} (a+1)^{3d+1}C_{f,c}^{a+1}\sum_{\und l\in\cL_{a}}4^{-L_{a}}\sum_{\und m\in\cM_{\und l,a}}\left( \frac{4R}{s}\right)^{M_{a}}(M_{a})^{3d}\\
    &\qquad\times\sum_{t=1}^{a}\;\sum_{\und k^{(t)}\in\cK_{l_t,m_t}} \;\sum_{\und j^{(t)}\in\cJ_{\und k^{(t)}}}\frac{1}{(k_1^{(1)}-1)!\dots (k_{l_{a}}^{(a)}-1)! }\nonumber  \\
    &\leq  \delta^2\kappa^{(0)}C_{f,d}\left(\frac{s^d}{\rho^d}+1\right)\big\langle N_{|x|<\rho}\big\rangle_0\sum_{a=0}^{n-1} (a+1)^{3d+1}C_{f,c}^{a+1}\sum_{\und l\in\cL_{a}}4^{-L_{a}}\sum_{\und m\in\cM_{\und l,a}}(M_{a})^{3d}\left( \frac{12eR}{s}\right)^{M_{a}}\label{very ugly 2}
\end{align}
where we used the same procedure as in \eqref{simplified sum}. Now we notice, since $M_{a}>1$, that there exists a constant $C_d$ such that

\begin{equation*}
    (M_{a})^{3d}\leq C_d(3d)^{M_{a}}.
\end{equation*}
\bigskip
This allows us to bound the sums in \eqref{very ugly 2} by

\begin{align*}
  \sum_{a=0}^{n-1}& (a+1)^{3d+1}C_{f,c}^{a+1}\sum_{\und l\in\cL_{a}}4^{-L_{a}}\sum_{\und m\in\cM_{\und l,a}}(M_{a})^{3d}\left( \frac{12eR}{s}\right)^{M_{a}}\\
  &\le C_d\sum_{a=0}^{n-1} (a+1)^{3d+1}C_{f,c}^{a+1}\sum_{\und l\in\cL_{a}}4^{-L_{a}}\sum_{\und m\in\cM_{\und l,a}}\left( \frac{36edR}{s}\right)^{M_{a}}\\
    &\le C_d\sum_{a=0}^{n-1} (a+1)^{3d+1}C_{f,c}^{a+1}\left(1- \frac{36edR}{s}\right)^{-a}\sum_{\und l\in\cL_{a}}\left( \frac{9edR}{s}\right)^{L_{a}}\\  
    &\le C_d\sum_{a=0}^{n-1} (a+1)^{3d+1}C_{f,c}^{a+1}\left(1- \frac{36edR}{s}\right)^{-a}\left( \frac{s}{9edR}-1\right)^{-a},
\end{align*}
where the last two inequalities hold thanks to \eqref{geometric series}, if $s>36edR$.

 Using again  \lemref{N and ASLTO} we obtain the following bound

 \begin{align*}
    C_{f,c}s\langle N_{f,ts}\rangle_0\leq C_{f,d}\delta^2\left(\frac{s^d}{\rho^d}+1\right)\big\langle N_{|x|<\rho}\big\rangle_0,
\end{align*}
 which concludes the proof.
 
 \end{proof}

We show that, for $s$ big enough, almost all the terms in  \corref{cor bound on integral before limit} vanish in the limit $n\to\infty$. This way  we are able to derive a bound that only depends on the number operator restricted on the region $\{|x|<\rho\}$, at the initial time.

 \begin{corollary}[Thermodynamically stable bound]\label{bound on int with only initial time}
     Consider an initial state such that, for all $ x\in\Lambda$, $\delta^{-1}\le \omega(n_x)\le\delta$ for some positive constant $\delta$. For any $R>0$, range of the interactions, $c>\kappa$, $v'>v$, and $f\in\cF$ there exist two constants $C_{f,d,c}$ and $C^1_{f,d,c,R}$, such that for every $s>C^1_{f,d,c,R}$ the following holds

     \begin{align*}
    \int_0^t \langle\: N_{f',rs}\rangle_r dr  \leq C_{f,c,d}\delta^2\left(\frac{s^d}{\rho^d}+1\right)\big\langle N_{|x|<\rho}\big\rangle_0,
    \end{align*}
 for any $b>0$ and $\rho=v't+b$.    
 \end{corollary}

 \begin{proof}
     Because of the hypothesis,  \corref{cor bound on integral before limit} holds and our goal is to understand the n dependency in each term of the left hand side of inequality \eqref{inequality before limit}. We start with the first summand

     \begin{align*}
       \kappa^{(0)}C_{f,c}t\left(\frac{C_{f,c}R }{s}\right)^n,  
     \end{align*}
     which goes to zero in the limit $n\to\infty$ if and only if 
     \begin{align}
         s>C_{f,c}R. \label{s cond 1}
     \end{align}
     The third term
     \begin{align*}
        C_{f,c}\kappa^{(0)}R  t \big\langle N\big\rangle_0\Big(\left(\frac{4R}{s}\right)^n\frac{1}{n!}+\left(\frac{4eR}{s}\right)^{n}\Big), 
     \end{align*}
     decays if 
     \begin{align}
         s>4eR.\label{s cond 2}
     \end{align}
     Now we focus on the fourth term
     \begin{align*}
       (\kappa^{(0)})^2 Rt\left(\frac{4R}{s}\right)^n\big\langle N\big\rangle_0\sum_{a=1}^nC_{f,c}^a\left(1-\frac{R}{s}\right)^{-a}\left(\frac{4s}{3eR}-1\right)^{-a},  
     \end{align*}
     which decays if $s>4R$ and
    \begin{align}
   & C_{f,c}\left(1-\frac{R}{s}\right)^{-1}\left(\frac{4s}{3eR}-1\right)^{-1}<1\nonumber\\
   \iff\quad & \left(1-\frac{R}{s}\right)\left(\frac{4s}{3eR}-1\right)>C_{f,c}\nonumber\\
   \iff\quad &\frac{4s}{3eR}-\frac{4}{3e}+\frac{R}{s}-1 >C_{f,c}\nonumber\\
   \Leftarrow \qquad &\frac{4s}{3eR} >C_{f,c}+\frac{4}{3e}+1\nonumber\\
   \iff \quad & s>\frac{3eR}{4}\left(C_{f,c}+\frac{4}{3e}+1\right).\label{s cond 3}
\end{align}
Analogously the fifth term
\begin{align*}
    (\kappa^{(0)})^2R t\left(\frac{8R}{s}\right)^n \big\langle N\big\rangle_0\sum_{a=1}^n C_{f,c}^a\left(1-\frac{R}{2s}\right)^{-a}\left(\frac{8s}{3eR}-1\right)^{-a}
\end{align*}
goes to zero if $s>8R$ and
\begin{align}
      & C_{f,c}\left(1-\frac{R}{2s}\right)^{-1}\left(\frac{8s}{3eR}-1\right)^{-1}<1\nonumber\\
     \iff \quad & \left(1-\frac{R}{2s}\right)\left(\frac{8s}{3eR}-1\right)>C_{f,c}\nonumber\\
     \iff\quad & \frac{8s}{3eR}+\frac{R}{2s}-\frac{4}{3e}-1>C_{f,c}\nonumber\\
     \Leftarrow\qquad &  \frac{8s}{3eR}>C_{f,c}+\frac{4}{3e}+1\nonumber\\
     \iff\quad & s>\frac{3eR}{8}\left(C_{f,c}+\frac{4}{3e}+1\right).\label{s cond 4}
\end{align}
Since there exists a constant $C_d$ such that $a^{3d}\leq C_d(3d)^a$, we notice that we can control the last term by
\begin{align*}
    \kappa^{(0)}&C_{f,d}\delta^2\left(\frac{s^d}{\rho^d}+1\right)\big\langle N_{|x|<\rho}\big\rangle_0\sum_{a=1}^na^{3d+1} C_{f,c}^a\left(1- \frac{36edR}{s}\right)^{-a+1}\left(\frac{s}{9edR}-1\right)^{-a+1}\\
    \leq&\kappa^{(0)}C_{f,d}\delta^2\left(\frac{s^d}{\rho^d}+1\right)\big\langle N_{|x|<\rho}\big\rangle_0\sum_{a=1}^n(3d+1)^a C_{f,c}^a\left(1- \frac{36edR}{s}\right)^{-a+1}\left(\frac{s}{9edR}-1\right)^{-a+1}
\end{align*}
we notice that, in the limit $n\to\infty$, it converges to
\begin{align*}
    C_{f,d}\delta^2\left(\frac{s^d}{\rho^d}+1\right)\big\langle N_{|x|<\rho}\rangle_0
\end{align*}
if the following holds
\begin{align}
     & (3d+1)C_{f,c}\left(1- \frac{36edR}{s}\right)^{-1}\left(\frac{s}{9edR}-1\right)^{-1}<1\nonumber\\
    \iff \quad & \left(1- \frac{36edR}{s}\right)\left(\frac{s}{9edR}-1\right)>C_{f,c,d}\nonumber\\
      \iff \quad & \frac{s}{9edR}+\frac{36edR}{s}-5>C_{f,c,d}\nonumber\\
      \Leftarrow\qquad & \frac{s}{9edR}>C_{f,c,d}\nonumber\\
      \iff \quad & s>RC_{f,c,d}.\label{s cond 5}
\end{align}
If $s$ satisfies \eqref{s cond 1}, \eqref{s cond 2}, \eqref{s cond 3}, \eqref{s cond 4}, and \eqref{s cond 5}, as $n$ goes to infinity, all those five terms go to zero and we obtain the desired inequality.
 \end{proof}
 
 We notice that if we allow for infinite range interactions, even very fast decaying ones, the term $\kappa^{(n)}$ would grow faster than exponential in n. This would make the terms diverge in the limit  $n\to\infty$.

\subsection{Main result}\label{true main res}
 In this section we will prove  \thmref{CT main}.
 
 \subsubsection{Proof for ball regions}
 We first show that the expectation of the number operator on the
region $\{|x| <b\}$ in the state evolved under the action of the Hamiltonian is controlled by the number of
particles initially in the region $\{|x| < \rho\}$. Then we conclude the proof generalizing this result to any bounded set $X$. 
 \begin{lemma}[Bound on particle transport velocity for ball regions]\label{main with balls}
   Consider an initial state such that, for all $ x\in\Lambda$, $\delta^{-1}\le\omega(n_x)\le\delta$ for some positive constant $\delta$. For any range of the interactions $R>0$, $c>\kappa$, and $v'>v$,  there exists a constant $C_{d,c,R}$ such that for any $b>0$, $\rho\ge b+v't$ it holds 
  \begin{align}
    \big\langle N_{|x|<b}\big\rangle_t  \leq  C_{d,c,R}\delta^2\big\langle N_{|x|<\rho}\big\rangle_0.
\end{align}
 \end{lemma}
 \begin{proof}
     By  \thmref{theorem bound on integral of exp of N}, we know that for any $n\in\mathcal{N}$ and $f\in\cF$ inequality \eqref{lemma before recursion} holds.

     If we drop the term $\frac{\kappa-v}{s}\int_0^t \langle\: N_{f',rs}\rangle_r dr$, bound $\langle N_{f,ts}\rangle_0$ using \eqref{ineq N Nf at time 0} and applying  \corref{bound on int with only initial time} to $\int_0^t \langle\: N_{f',rs}\rangle_r dr$ we obtain

     \begin{align}
    \langle N_{f,ts}\rangle_t\leq  \langle &N_{f,ts}\rangle_0\nonumber\\
    &+C\sum_{l=1}^{n-1}\sum_{m=l}^n \sum_{\substack{1\leq k_i,\,\forall i\in[1,l]\\k_1+\dots+k_l=m}} \sum_{\substack{0\leq j_1\leq k_1 \\ \forall i\in[1,l]}}\frac{4^{m-l}\cdot s^{-m-1}}{(k_1-1)!\dots (k_l-1)! }\kappa^{(m)}\int_0^t\big\langle N_{f'_{J_l,rs}}\big\rangle_rdr  \nonumber  \\
    &+C_f4^n\Big(\frac{1}{n!}+\sum_{l=2}^{n-1}\left(\frac{e}{2}\right)^l\sum_{m=l-1}^n  2^m\Big)s^{-n-1}\kappa^{(n+1)}  t \big\langle N\big\rangle_0\nonumber\\
    \leq\;  &C_{f,d}\delta^2\left(\frac{s^d}{\rho^d}+1\right)\big\langle N_{|x|<\rho}\big\rangle_0\nonumber\\
    &+C_{f,d}\delta^2\left(\frac{s^d}{\rho^d}+1\right)\sum_{l=1}^{n-1}\sum_{m=l}^n \sum_{\substack{1\leq k_i,\,\forall i\in[1,l]\\k_1+\dots+k_l=m}} \sum_{\substack{0\leq j_1\leq k_1 \\ \forall i\in[1,l]}}\frac{4^{m-l}\cdot s^{-m-1}}{(k_1-1)!\dots (k_l-1)! }\kappa^{(m)}t\big\langle N_{|x|<\rho}\big\rangle_0   \nonumber \\
    &+C_f4^n\Big(\frac{1}{n!}+\sum_{l=2}^{n-1}\left(\frac{e}{2}\right)^l\sum_{m=l-1}^n  2^m\Big)s^{-n-1}\kappa^{(n+1)}  t \big\langle N\big\rangle_0\nonumber
\end{align}

By the same argument as in  \corref{bound on int with only initial time} we see that as $n$ goes to infinity the last two terms go to zero and the inequality becomes
\begin{align}
    \langle N_{f,ts}\rangle_t\leq C_{f,d}\delta^2\left(\frac{s^d}{\rho^d}+1\right)\big\langle N_{|x|<\rho}\big\rangle_0.\nonumber
\end{align}

If we choose $s=C_R\rho$, for some constant $C_R>0$ such that conditions \eqref{s cond 1}, \eqref{s cond 2}, \eqref{s cond 3}, \eqref{s cond 4}, and \eqref{s cond 5} are fulfilled, we obtain the following bound:

\begin{align}
    \langle N_{f,ts}\rangle_t\leq C_{f,d,c,R}\big\langle N_{|x|<\rho}\big\rangle_0.\nonumber
\end{align}

Now we apply \eqref{ineq N Nf at time t} on the left hand side of this inequality and we conclude the proof.
 \end{proof}

\subsubsection{Proof of Theorem \ref{CT main}}\label{secPfMain}

\begin{proof}

Since $X$ is a bounded set there exists $r>0$ such that $X\subset B_r$. We define 
\begin{align*}
    \eta_0:=\inf\{r>0 \quad s.t\quad X\subset B_r\}.
\end{align*}

This implies that
\begin{equation*}
    \omega_t(N_X)\le \omega_t(N_{B_{\eta_0}}).
\end{equation*}
By  \lemref{main with balls} we know that for any $\eta>vt$ there exists a constant $C\equiv C_{d,c,R}$ such that

\begin{equation*}
    \omega_t(N_{B_{\eta_0}})\le \delta^2 C \omega(N_{B_{\eta_0+\eta}}).
\end{equation*}
Since we assumed controlled density in the initial state
\begin{equation*}
    C \omega(N_{B_{\eta_0+\eta}})\le C{\delta}|B_{\eta_0+\eta}|=C{\delta}(\eta_0+\eta)^d\le C\delta2^d \eta^d\le C\delta |B_\eta|\le C \delta^2\omega(N_{B_\eta}),
\end{equation*}
if we assume $\eta>\eta_0$. We point out that the constant depends on the dimension $d$, the range $R$ and $c$.
By definition of $X_\eta$, there always exists a ball $B_\eta$ contained in $X_\eta$, which implies
\begin{equation*}
    \omega(N_{B_{\eta}})\le\omega (N_{X_\eta})
\end{equation*}
and this concludes the proof.
\end{proof}

\end{document}